 \documentclass[12pt, draftclsnofoot, onecolumn]{IEEEtran}
\usepackage{amsmath,graphicx,cite}
\usepackage{wrapfig,lipsum}

\everymath{\displaystyle}
\title{Interference Coordination via Power Domain Channel Estimation}

\author
  {Chao Zhang $^\dagger$, Vineeth S. Varma $^*$, Samson Lasaulce  $^\dagger$, and Rapha$\ddot{\mathrm{e}}$l Visoz  $^\ddagger$
	\thanks{$^\dagger$  L2S (CNRS-CentraleSupelec-Univ. Paris Sud), Gif-sur-Yvette, France.}
	\thanks{$^*$ CRAN (CNRS-Univ. of Lorraine), Nancy, France.}
\thanks{$^\ddagger$ Orange Labs, Issy-les-Moulineaux, France.}
\thanks{The material in this paper was presented in part at the 2015 EUSIPCO Conference \cite{varma-eusipco-2015}.}}

\usepackage[mathcal]{euscript}
\usepackage{amssymb,amsmath}
\usepackage{algpseudocode}
\usepackage[]{algorithm2e}
\usepackage{verbatim}
\usepackage{graphicx}
\usepackage{longtable}
\usepackage{rotating}
\usepackage{multirow}
\usepackage{tabularx}
\usepackage{bigints}
\usepackage{tikz}
\usepackage{multirow} 
\usetikzlibrary{arrows,backgrounds,snakes}
\usepackage{amsmath,amstext,amsfonts,amssymb}
\usepackage{amsbsy}
\usepackage{amsthm}
\usepackage{diagbox}
\usepackage{mathabx}
\usepackage{color}
\usepackage{float}
\usetikzlibrary{arrows,shapes,positioning,shadows,trees}
\usetikzlibrary{decorations.pathreplacing}

\usepackage{float}
\usepackage{caption}
\usepackage{subcaption}

\tikzset{
  basic/.style  = {draw, text width=2cm, drop shadow, font=\sffamily, rectangle},
  root/.style   = {basic, rounded corners=2pt, thin, align=center,
                   fill=green!30},
  level 2/.style = {basic, rounded corners=6pt, thin,align=center, fill=white!60,
                   text width=10em},
  level 3/.style = {basic, thin, align=left, fill=pink!60, text width=6.5em}
}

\definecolor{darkred}{rgb}{0.5,0.0,0.0}
\definecolor{darkblue}{rgb}{0.2,0.1,0.6}

\newcommand{\ul}{\underline}

\newcommand{\Gm}{\textbf{G}}

\newtheorem{theorem}{Theorem}[section]

\newtheorem{proposition}[theorem]{Proposition}
\begin{document}

\maketitle
\vspace{-0.15cm}
\vspace{-0.15cm}
\vspace{-0.15cm}
\vspace{-0.15cm}
\vspace{-0.15cm}
\begin{abstract}
A novel technique is proposed which enables each transmitter to acquire global channel state information (CSI) from the sole knowledge of individual received signal power measurements, which makes dedicated feedback or inter-transmitter signaling channels unnecessary. To make this possible, we resort to a completely new technique whose key idea is to exploit the transmit power levels as symbols to embed information and the observed interference as a communication channel the transmitters can use to exchange coordination information. Although the used technique allows any kind of {low-rate} information to be exchanged among the transmitters, the focus here is to exchange local CSI. The proposed procedure also comprises a phase which allows local CSI to be estimated. Once an estimate of global CSI is acquired by the transmitters, it can be used to optimize any utility function which depends on it. While algorithms which use the same type of measurements such as the iterative water-filling algorithm (IWFA) implement the sequential best-response dynamics (BRD) applied to individual utilities, here, thanks to the availability of global CSI, the BRD can be applied to the sum-utility. Extensive numerical results show that significant gains can be obtained and, this, by requiring no additional online signaling.

\end{abstract}
\vspace{-0.15cm}
\section{Introduction}

Interference networks are wireless networks which are largely distributed decision-wise or information-wise. In the case of distributed power allocation over interference networks with multiple bands, the iterative water-filling algorithm (IWFA) is considered to be one of the well-known state-of-the art distributed techniques \cite{yu-jsac-2002}\cite{scutari-tsp-2009}\cite{mertikopolous-jsac-2012}. IWFA-like distributed algorithms have at least two attractive features: they only rely on local knowledge e.g., the individual signal-to-interference plus noise ratio (SINR), making them distributed information-wise; the involved computational complexity is typically low. On the other hand, one drawback of IWFA and many other distributed iterative and learning algorithms (see e.g., \cite{rose-cm-2011}\cite{lasaulce-book-2011}) is that convergence is not always ensured \cite{mertikopolous-jsac-2012} and, when converging, it leads to a Nash point which is  globally inefficient. 

One of the key messages of the present paper is to show that it is possible to exploit the available feedback signal more efficiently than IWFA-like distributed algorithms do. In the exploration phase\footnote{IWFA operates over a period which is less than the channel coherence time and it does so in two steps: an exploration phase during which the transmitters update in a round robin manner their power allocation vector; an exploitation phase during which the transmitters keep their power vector constant  at the values obtained at the end of the exploration phase. {As for IWFA, unless mentioned otherwise, we will assume the number of time-slots of the exploitation phase to be much larger than that of the exploration phase, making the impact of the exploration phase on the average performance negligible.}

}, instead of using local observations (namely, the individual feedback) to allow the transmitters to converge to a Nash point, one can use them to acquire global channel state information (CSI). This allows coordination to be implemented, and more precisely global performance criteria or network utility to be optimized during the exploitation phase. As for complexity, it has to be managed by a proper choice of the network utility function which has to be maximized.

To obtain global CSI, one of the key ideas of this paper is to exploit the transmit power levels as information symbols and to exploit the interference observed to decode these information symbols. In the literature of power control and resource allocation, there exist papers where the observation of interference is exploited to optimize a given performance criterion. In this respect, an excellent monograph on power control is \cite{chiang-book-2007}. Very relevant references include \cite{stanczak2007} and \cite{schreck2015}. In \cite{stanczak2007}, optimal power control for a reversed network (receivers can transmit) is designed, in which the receiver uses the interference to estimate the cross channel, assuming perfect exchange of information between the transmitters. In \cite{schreck2015}, the authors estimate local CSI from the received signal but in the signal domain and in a centralized setting. To the best to the authors' knowledge, there is no paper where the interference measurement is exploited as a communication channel the transmitters can utilize to exchange information or local CSI (namely, the channel gains of the links which arrive to a given receiver), as is the case under investigation. In fact, we provide a complete estimation procedure which relies on the sole knowledge of the individual received signal strength indicator (RSSI). The proposed approach is somewhat related to the Shannon-theoretic work on coordination available in \cite{larousse}\cite{larrousse-tit-sub-2015}, which concerns two-user interference channels when one master transmitter knows the future realizations of the global channel state. 


{It is essential to insist on the fact that the purpose of the proposed estimation scheme is not to compete with conventional estimation schemes such as \cite{caire-tit-2010} (which are performed in the signal-domain), but rather, to evaluate the performance of an estimation scheme that solely relies on information available in the power-domain. Indeed, one of the key results of the paper is to prove that global CSI (without phase information)} can be acquired from the \textbf{sole} knowledge of a \textbf{given} feedback which is the SINR or RSSI feedback. The purpose of such a feedback is generally to adjust the power control vector or matrix but, to our knowledge, it has not been shown that it also allows global CSI to be recovered, and additionally, at every transmitter. This sharply contrasts with conventional channel estimation techniques which operate in the \textbf{signal} domain and use a \textbf{dedicated} channel for estimation.

The main contributions and novelty of this work are as follows:\\
$\blacktriangleright$ {We introduce the important and novel idea of communication in the power domain, i.e., by encoding the message on the transmit power and decoding by observing the received signal strength. This can be used in fact to \textbf{exchange any kind of low-rate information} and not only CSI.}\\
$\blacktriangleright$ This allows interfering transmitters to exchange information \textbf{without requiring the presence of dedicated signaling channels} (like direct inter-transmitter communication), which may be unavailable in real systems (e.g., in conventional Wifi systems or heterogeneous networks).\\
$\blacktriangleright$  Normal (say high-rate) communication can be done even during the proposed learning phase with a sub-optimal power control, i.e., communication during the learning time in the proposed scheme is similar to communication in the convergence time for algorithms like IWFA.\\
$\blacktriangleright$ We propose a way to both learn and exchanged the local CSI. Global CSI is acquired at every transmitter by observing the RSSI feedback. \\
$\blacktriangleright$ The proposed technique accounts for the presence of various noise sources which are non-standard and affect the RSSI measurements (the corresponding modeling is provided in Sec. II). By contrast, apart from a very small fraction of works (such as \cite{mertikopolous-jsac-2012}\cite{anandkumar-tvt-2011}\cite{coucheney-isit-2014}), IWFA-like algorithms assume noiseless measurements.\\
$\blacktriangleright$ We conduct a detailed performance analysis to assess the benefits of the proposed approach for the exploitation phase, which aims at optimizing the sum-rate or sum-energy-efficiency. As (imperfect) global CSI is available, globally efficient solutions become attainable. The proposed work can be extended in many respects; the main extensions are marked as ($\star$).

\vspace{-0.15cm}
\section{Problem statement and proposed technique general description}
\label{sec:sysm}

\textbf{Channel and communication model}: The system under consideration comprises $K \geq 2$ pairs of interfering transmitters and receivers; each transmitter-receiver pair will be referred to as a user. Our technique directly applies to the multi-band case, and this has been done in the numerical section. In particular, we assess the performance gain which can be obtained with respect to the IWFA. However, for the sake of clarity and ease of exposition, we focus on the single-band case, and explain in the end of Sec. IV, the modifications required to treat the multi-band case. From this point on, we will therefore assume the single-band case unless otherwise stated.

In the setup under study, the quantities of interest for a transmitter to control its power are given by the channel gains. The channel gain of the link between Transmitter $i \in \{1,...,K\}$ and Receiver $j \in \{1,...,K\}$ is denoted by $g_{ij}= |h_{ij}|^2$, where $h_{ij}$ may typically be the realization of a complex Gaussian random variable, if Rayleigh fading is considered. In several places in this paper we will use the $K \times K$ \textit{channel matrix} $\Gm$ whose entries are given by the channel gains $g_{ij}$, $i$ and $j$ respectively representing the row and column indices of $\Gm$.  Each channel gain is assumed to obey a classical block-fading variation law. More precisely, channel gains are assumed to be constant over each transmitted data frame. A \textit{frame} comprises $T_{\mathrm{I}}+T_{\mathrm{II}}+T_{\mathrm{III}}$ consecutive time-slots where $T_m \in \mathbb{N}$, $m\in \{\mathrm{I}, \mathrm{II}, \mathrm{III}\}$, corresponds to the number of time-slots of Phase $m$ of the proposed procedure; these phases are described further. Transmitter $i$, $i\in\{1,...,K\}$, can update its power from time-slot to time-slot. The corresponding power level is denoted by $p_i$ and is assumed to be subject to power limitation as: $0 \leq p_i \leq P_{\max}$. The $K-$dimensional column vector formed by the transmit power levels will be  denoted by $\ul{p}=(p_1,...,p_K)^{\mathrm{T}}$, $\mathrm{T}$ standing for the transpose operator.\\ 
\vspace{-0.15cm}
\textbf{Feedback signal model}: We assume the existence of a feedback mechanism which provides each transmitter, an image or noisy version of the power received at its intended receiver for each time-slot. The power at Receiver $i$ on time-slot $t$ is expressed as
\vspace{-0.15cm}
\begin{equation}\label{eq:real-omega}
\omega_{i}(t)  =  g_{ii}  p_{i}(t) + \sigma^2+  \sum_{j \neq i} g_{ji}  p_{j}(t). 
\end{equation}
\vspace{-0.15cm}
where $\sigma^2$ is the receive noise variance and $p_i(t)$ the power of Transmitter $i$ on time-slot $t$. We assume that the following procedure is followed by the transmitter-receiver pair. Receiver $i$: measures the received signal (RS) power $\omega_i(t)$ at each time slot and quantizes it with $N$ bits (the \textit{RS power quantizer} is denoted by $\mathcal{Q}_{\mathrm{RS}}$); sends the quantized RS power $\widehat{\omega}_i(t)$ as feedback to Transmitter $i$ through a noisy feedback channel. After quantization, we assume that for all $i \in \{1,...,K\}$, $\widehat{\omega}_i(t) \in W$, where $W = \{ \mathrm{w}_1 , \mathrm{w}_2,\dots, \mathrm{w}_M \}$ such that $0 \leq \mathrm{w}_1 < \mathrm{w}_2 < \dots < \mathrm{w}_M $ and $M=2^N$. Transmission over the feedback channel and the dequantization operation are represented by a discrete memoryless channel (DMC) whose \textit{conditional probability} is denoted by $\Gamma$. The distorted and noisy version\footnote{Note that, for the sake of clarity, it is assumed here that the RS power quantizer and DMC are independent of the user index, but the proposed approach holds in the general case.} of $\omega_{i}(t)$, which is available at Transmitter $i$, is denoted by ${\widetilde{\omega}}_i(t) \in W$; the quantity ${\widetilde{\omega}}_i(t)$ will be referred to as the \textit{received signal strength indicator (RSSI)}. With these notations, the probability that Transmitter $i$ decodes the symbol $\mathrm{w}_\ell$ given that Receiver $i$ sent the quantized RS power $\mathrm{w}_k$ equals $\Gamma(\mathrm{w}_\ell | \mathrm{w}_k )$. 


In contrast with the vast majority of works on power control and especially those related to the IWFA, we  assume the feedback channel to be noisy. Note also that these papers typically assume SINR feedback whereas the RSSI is considered here. The reasons for this is fourfold: 1) if Transmitter $i$ knows $p_i (t)$, $g_{ii} (t)$, and has SINR feedback, this amounts to knowing its RS power since $\omega_{i}(t) =  \displaystyle{g_{ii}  p_{i} (t)  \left(1+ \frac{1}{\mathrm{SINR}_{i}  (t)} \right)}$ where $\mathrm{SINR}_{i}(t)   = \frac{ g_{ii}  p_{i}(t)}{\displaystyle{ \sigma^2+\sum_{j \neq i} g_{ji}  p_{j}(t) }}$; 2) Assuming an RS power feedback is very relevant in practice since some existing wireless systems exploit the RSSI feedback signal (see e.g., \cite{sesia-book-2009}); 3) The SINR is subject to higher fluctuations than the RS power, which makes SINR feedback less robust to distortion and noise effects and overall less reliable; 4) As a crucial technical point, it can be checked that using the SINR as the transmitter observation leads to complex estimators \cite{patentpm}, while the case of RS power observations leads to a simple and very efficient estimation procedure, as shown further in this paper.\\
\vspace{-0.15cm}
Note that, here, it is assumed that the RS power is quantized and then transmitted through a DMC, which is a reasonable and common model for wireless communications. Another possible model for the feedback might consist in assuming that the receiver sends directly received signal power over an AWGN channel; depending on how the feedback channel gain fluctuations may be accounted for, the latter model might be more relevant and would deserve to be explored as well ($\star$).

%

\textbf{Proposed technique general description}: The general power control problem of interest consists in finding, for each realization of the channel gain matrix $\Gm$, a power vector which maximizes a network utility of the form $u(\ul{p};\Gm)$. For this purpose, each transmitter is assumed to have access to the realizations of its RSSI over a frame. One of the key ideas of this paper is to exploit the transmit power levels as information symbols and exploit the observed interference (which is observed through the RSSI or SINR feedback) for inter-transmitter communication. The corresponding implicit communication channel is exploited to acquire global CSI knowledge namely, the matrix $\Gm$ and therefore to perform operations such as the maximization of $u(\ul{p};\Gm)$.
\begin{wrapfigure}{r}{0.5\textwidth}
   \begin{center}
\includegraphics[width=0.5\textwidth]{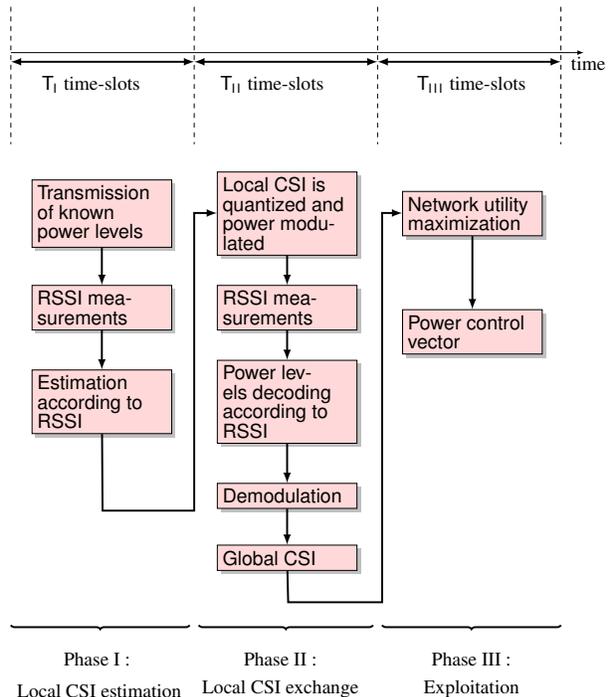}
  \end{center}
   \caption{\small The flowchart of the proposed scheme}
    \label{fig:flowc}
\end{wrapfigure} 
The process of achieving the desired power control vector is divided into three phases (see Fig.~\ref{fig:flowc}). In Phase I, a sequence of power levels which is known to all the transmitters is transmitted (similar to a training sequence in classical channel estimation but in the power domain), and Transmitter $i$ estimates its own channel gains (i.e., $g_{1i}, g_{2i},...,g_{Ki}$) by exploiting the noisy RSSI feedback; we refer to the corresponding channel gains as \textit{local CSI}. In Phase II, each transmitter informs the other transmitters about its local CSI by using power modulation. By decoding the modulated power, each transmitter can estimate the channel gains of the other users and thus, at the end of Phase II each transmitter has its own estimate of the \textit{global CSI} $\Gm$; the situation where transmitters have a non-homogeneous or different knowledge of global CSI is referred to as a distributed CSI scenario in \cite{Kerret}. In Phase III, each transmitter can then exploit global CSI to maximize (possibly in a sub-optimal manner) the network utility of interest. In the numerical part, we make specific and classical choices for the network utility namely, we consider the network sum-rate and network sum-energy-efficiency.

%

\section{Phase I: Local CSI estimation in the power domain}
%
%
%
%

Phase $\mathrm{I}$ comprises $T_{\mathrm{I}}$ time-slots. The aim of Phase $\mathrm{I}$ is to allow Transmitter $i$, $i\in\{1,...,K\}$, to acquire local CSI from the $T_{\mathrm{I}}$ observations $\widetilde{\omega}_i(1),...,\widetilde{\omega}_i(T_{\mathrm{I}})$ which are available thanks to the feedback channel between Receiver $i$ and Transmitter $i$. Obviously, if local CSI is already available e.g., because another estimation mechanism is available, Phase $\mathrm{I}$ can be skipped and one can directly proceed with the local CSI exchange among the transmitters namely, performing Phase $\mathrm{II}$. 

For every time-slot of Phase $\mathrm{I}$, each transmitter transmits at a prescribed power level which is assumed to be known to all the transmitters. One of the key observations we make in this paper is that, when the channel gains are constant over several time-slots, it is possible to recover local CSI from the RSSI or SINR; this means that, as far as power control is concerned, there is no need for additional signaling from the receiver for local CSI acquisition by the transmitter. Thus, the sequences of power levels in Phase $\mathrm{I}$ can be seen as training sequences. Technically, a difference between classical training-based estimation and Phase $\mathrm{I}$ is that estimation is performed in the power domain and over several time-slots and not in the symbol domain (symbol duration is typically much smaller than the duration of a time-slot) within a single time-slot. Also note that working in the symbol domain would allow one to have access to $h_{ij}$ but the phase information on the channel coefficients is irrelevant for the purpose of maximizing a utility function of the form $u(\ul{p};\Gm)$. Another technical difference stems from the fact that the feedback noise is not standard, which is commented more a little further.

By denoting $ (p_i(1),...,p_i(   T_{\mathrm{I}}  ))$,   $i \in \{1,...,K\}$, the sequence of training power levels used by Transmitter $i$, the following training matrix can be defined:
\begin{equation}
\mathbf{P}_{\mathrm{I}}= 
\left(
\begin{array}{ccc}
 p_{1} (1) &  \dots & p_{K} (1) \\
  \vdots &     \vdots & \vdots \\
 p_{1} (T_{\mathrm{I}}  ) & \dots & p_{K} ( T_{\mathrm{I}} )
\end{array}
\right).
\end{equation}
With the above notations, the \textit{noiseless RS power vector} $\ul{\omega}_i = (\omega_i(1),...,\omega_i(T_{\mathrm{I}} ))^\mathrm{T}$ can be expressed as:
\begin{equation}
\underline{\omega}_i=  \mathbf{P}_{ \mathrm{I}   }    \ul{g_{i}} + \sigma^2 \ul{1}.
\end{equation}
where $\underline{g}_i = (g_{1i},..,g_{Ki})^\mathrm{T}$ and $\ul{1} = (1,1,...,1)^\mathrm{T}$. 

To estimate the local CSI $\underline{g}_i$ from the sole knowledge of the \textit{noisy RS power vector or RSSI}  $\underline{\widetilde{\omega}}_i$ we propose to use the least-squares (LS) estimator in the power domain (PD), abbreviated as LSPD, to estimate the local CSI as:
\vspace{-0.15cm}
\begin{equation}
\ul{{\widetilde{g}}}_i^{\mathrm{LSPD}}=
 \left(   \mathbf{P}_{ \mathrm{I} }^{\mathrm{T}}
  \mathbf{P}_{ \mathrm{I} }  \right)^{-1} \mathbf{P}_{ \mathrm{I} }^{\mathrm{T}}  \left(
\underline{\widetilde{\omega}}_{i}  - \sigma^2 \ul{1} \right).
\label{eq:LS}
\end{equation}
where $\sigma^2$ is assumed to be known from the transmitters since it can always be estimated through conventional estimation procedures (see e.g., \cite{lasaulcevtc}). Using the LSPD estimate for local CSI therefore assumes that the training matrix $\mathbf{P}_{ \mathrm{I} }$ is chosen to be pseudo-invertible. A necessary condition for this is that the number of time-slots used for Phase $\mathrm{I}$ verifies: $T_{\mathrm{I}}  \geq K$. Using a diagonal training matrix allows this condition to be met and to simplify the estimation procedure. 

It is known that the LSPD estimate may coincide with the maximum likelihood (ML) estimate. This holds for instance when the observation model of the form $\widetilde{\omega}_i =\omega_i+z$ where $z$ is an independent and additive white Gaussian noise. In the setup under investigation, $z$ represents both the effects of quantization and transmission errors over the feedback channels and does not meet neither the independence nor the Gaussian assumption. However, we have identified a simple and sufficient condition under which the LSPD estimate maximizes the \textit{likelihood} $P(\widetilde{\ul{\omega}}_i | \ul{g}_i)$. This is the purpose of the next proposition.
\vspace{-0.15cm}
\begin{proposition} Denote by $\mathcal{G}_{i}^{\mathrm{ML}} $ the set of ML estimates of $\ul{{g}}_i$, then we have
$$
\begin{array}{cl}
(i)  &  \mathcal{G}_i^{\mathrm{ML}}=  \arg \underset{\underline{g}_{i}}{\max} \overset{T_{\mathrm{I}}}{\underset{t=1}{\displaystyle\prod}}\Gamma\left(\widetilde{\omega}_{i}\left(t\right) \left|Q_{\mathrm{RS}}\left(\ul{e}_t^{\mathrm{T}}   \mathbf{P}_{ \mathrm{I} } \underline{g}_{i}+\sigma^2\right) \right. \right);\\ 
(ii) & \ul{\widetilde{g}}_i^{\mathrm{LSPD}} \in \mathcal{G}_{i}^{\mathrm{ML}} \mathrm{\>\>when\>for\> all\>} \ell,   \ \arg \underset{k}{\max} \:\:\Gamma(\mathrm{w}_\ell |  \mathrm{w}_k )   =   \ell; 
\end{array}
$$
where $\ul{e}_t$ is a column vector whose entries are zeros except for the $t^{\mathrm{th}}$ entry which equals $1$. 
\end{proposition}
\begin{proof}
See Appendix A.
\end{proof}
\vspace{-0.15cm}
The sufficient condition corresponding to $(ii)$ is clearly met in classical practical scenarios. Indeed, as soon as the probability of correctly decoding the sent quantized RS power symbol (which is sent by the receiver) at the transmitter exceeds $50\%$, the above condition is verified. It has to be noted that $\mathcal{G}_i^{\mathrm{ML}}$ is not a singleton set in general, which indicates that even if the LSPD estimate maximizes the likelihood, the set $\mathcal{G}_i^{\mathrm{ML}}$ will typically comprise a solution which can perform better e.g., in terms of mean square error.

If some statistical knowledge on the channel gains is available, it is possible to further improve the performance of the channel estimate.  Indeed, when the probability  of $\ul{g}_i$ is known it becomes possible (up to possible complexity limitations) to minimize the mean square error $\mathbb{E} \|\underline{\widehat{g}}_i-\underline{g}_i\|^2$. The following proposition provides the expression of the minimum mean square error (MMSE) estimate in the power domain (PD) . 
\vspace{-0.15cm}
\begin{proposition} Assume that $\forall i \in \{1,...,K\}$, $\ul{\widehat{\omega}}_i$ and  $\ul{\widetilde{\omega}}_i$ belong to the set ${\Omega} = \left\{\ul{\mathrm{w}}_1,..., \ul{\mathrm{w}}_{M^{T_{\mathrm{I}}}} \right\}$, where $\ul{\mathrm{w}}_1 = (\mathrm{w}_1,\mathrm{w}_1, ...,\mathrm{w}_1)^{\mathrm{T}}$, $\ul{\mathrm{w}}_2 = (\mathrm{w}_1, \mathrm{w}_1,...,\mathrm{w}_2)^{\mathrm{T}}$,..., $\ul{\mathrm{w}}_{M^{T_{\mathrm{I}}}} = (\mathrm{w}_M,\mathrm{w}_M, ...,\mathrm{w}_M)^{\mathrm{T}}$ (namely, vectors are ordered according to the lexicographic order and have $T_{\mathrm{I}}$ elements each). Define $\mathcal{G}_{m}$ as
\vspace{-0.15cm}
\begin{equation}
\mathcal{G}_{m}:=\left\{\underline{x} \in \mathbb{R}_{+}^K: \     Q_{\mathrm{RS}}\left(\mathbf{P}_{ \mathrm{I} } \underline{x}+\sigma^2\ul{1}\right)
=\ul{\mathrm{w}}_m\right\}.
\end{equation}
\vspace{-0.15cm}
Then the MMSE estimator in the power domain expresses as:
\begin{equation}\label{eq:MMSEPD}
\ul{\widetilde{g}}_{i}^{\mathrm{MMSEPD}}=\frac{{\displaystyle\sum_{m=1}^{M^{T_{\mathrm{I}}}}}\overset{T_{\mathrm{I}}}{\underset{t=1}{\prod}}\Gamma \left(\widetilde{\omega}_i(t)|\ul{\mathrm{w}}_m(t)\right) \displaystyle\int_{\mathcal{G}_{m}}\phi_i\left(\ul{g}_i\right)
\underline{g}_i{\mathrm{d} g_{1i}...\mathrm{d} g_{Ki}}}{{\displaystyle\sum_{m=1}^{M^{T_{\mathrm{I} }}}}\overset{T_{ \mathrm{I}  }}{\underset{t=1}{\prod}}\Gamma \left(\widetilde{\omega}_i(t)|\ul{\mathrm{w}}_m(t)\right) \displaystyle\int_{\mathcal{G}_{m}} \phi_i\left(\ul{g}_i\right){  \mathrm{d} g_{1i}...\mathrm{d} g_{Ki}  }},
\end{equation}
where $\phi_i$ represents the probability density function (p.d.f.) of $\ul{g}_i$ and $\ul{\mathrm{w}}_m(t)$ is the $t$-th element of $\ul{\mathrm{w}}_m$.  
\end{proposition}
\begin{proof}
See Appendix B.
\end{proof}
\vspace{-0.15cm}
In the simulation section (Sec. V), we will compare the LSPD and MMSEPD performance in terms of estimation SNR, sum-rate, and sum-energy-efficiency. While the MMSEPD estimate may provide a quite significant gain in terms of MSE over the LSPD estimate, it also has a much higher computational cost. Simulations reported in Sec. V will exhibit conditions under which choosing the LSPD solution may involve a marginal loss w.r.t. the MMSEPD solution e.g., when the performance is measured in terms of sum-rate. Therefore the choice of the estimator can be made based on the computation capability, the choice of utility for the system under consideration, or the required number of time-slots (MMSEPD allows for a number of time-slots which is less than $K$, whereas this is not possible for LSPD). Note that some refinements might be brought to the proposed estimator e.g., by using a low-rank approximation of the channel vector (see e.g., \cite{lasaulce-vtc-2001}), which is particularly relevant if the channel appears to possess some sparseness.
\vspace{-0.15cm}
\section{Phase II: Local CSI exchange in the power domain}
Phase $\mathrm{II}$ comprises $T_{\mathrm{II}}$ time-slots. The aim of Phase $\mathrm{II}$ is to allow Transmitter $i$, $i\in\{1,...,K\}$, to exchange its knowledge about local CSI with the other transmitters; the corresponding estimate will be merely denoted by $\ul{\widetilde{g}}_i = \left(\widetilde{g}_{1i},..., \widetilde{g}_{Ki} \right)^{\mathrm{T}}$, knowing that it can refer either to the LSPD or MMSEPD estimate. The proposed procedure is as follows and is also summarized in Fig.~3. Transmitter $i$ quantizes the information $\widetilde{\ul{g}}_i$ through a \textit{channel gain quantizer} called $\mathcal{Q}_{i}^{\mathrm{II}}$ and maps the obtained bits (through a modulator) into the sequence of power levels $\ul{p}_i^{\mathrm{II}} = (p_i(T_{\mathrm{I}}+1),..., p_i(T_{\mathrm{I}}+T_{\mathrm{II}}))^{\mathrm{T}}$. From the RSSI observations $\widetilde{\ul{\omega}}_j^{\mathrm{II}} = (\widetilde{\omega}_j(T_{\mathrm{I}}+1),..., \widetilde{\omega}_j(T_{\mathrm{I}}+T_{\mathrm{II}}))^{\mathrm{T}}$, Transmitter $j$ ($j\neq i$) can estimate (through a decoder) the power levels used by Transmitter $i$. To facilitate the corresponding operations, we assume that the used power levels in Phase $\mathrm{II}$ have to lie in  the \textit{reduced set} $\mathcal{P} =\{P_1,...,P_L \}$ with $\forall \ell \in \{1,...,L\}$, $P_\ell \in [0,P_{\max}]$. The estimate Transmitter $j$ has about the channel vector $\ul{g}_i$ will be denoted by $\widetilde{\ul{g}}_i^j = \left(\widetilde{g}_{1i}^j,...,\widetilde{g}_{Ki}^j \right)^{\mathrm{T}}$. The corresponding \textit{channel matrix estimate} is denoted by $\widetilde{\mathbf{G}}^j$. 
\begin{figure}
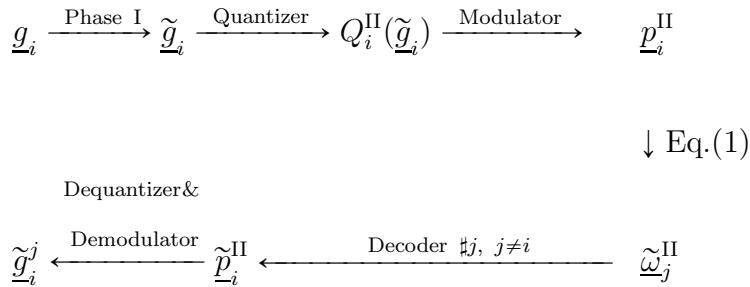
\label{fig:coder}
\begin{equation*}
\begin{array}{ll}
\ul{g}_i \xrightarrow{\mathrm{\ Phase} \ \mathrm{I}} \widetilde{\ul{g}}_i \xrightarrow{\mathrm{ \ Quantizer \ }} Q_{i}^{\mathrm{II}}(\widetilde{\ul{g}}_i) \xrightarrow{\mathrm{ \ Modulator \ }}  & \ul{p}_i^{\mathrm{II}}\\ 
& \:\\
&{{\downarrow  \mathrm{Eq. (1)}}}\\ 
 \widetilde{\ul{g}}_i^j   \xleftarrow{{\scriptsize \begin{split}\mathrm{Dequantizer}   \& \\ \mathrm{Demodulator}\end{split}}}     \widetilde{\ul{p}}_i^{\mathrm{II}}
 
  \xleftarrow{  \ \ \ \      \ \ \ \ \  \mathrm{Decoder} \ \sharp j, \ j\neq i \ \ \ \  \ \ \   }   & \widetilde{\ul{\omega}}_j^{\mathrm{II}}
\end{array}
\end{equation*}
{\caption{\small The figure summarizes the overall processing chain for the CSI}}
\end{figure}
In what follows, we describe the proposed schemes for the three operations required to exchange local CSI namely, quantization, power modulation, and decoding. The situation where transmitters have different estimates of the same channel is referred to as a distributed CSI scenario in \cite{Kerret}. Assessing analytically the impact of distributed CSI on the sum-rate or sum-energy-efficiency is beyond the scope of this paper but constitutes a very relevant extension of it ($\star$); only simulations accounting for the distributed CSI effect will be provided here.\\ 
\vspace{-0.15cm}
\vspace{-0.15cm}

It might be noticed that the communication scenario in Phase II is similar to the X-channel scenario in the sense that each transmitter wants to inform the other transmitters (which play the role of receivers) about its local CSI, and this is done simultaneously.  All the available results on the X-channel exploit the channel structure (e.g., the phase information) to improve performance (e.g., by interference alignment \cite{maddahali-tit-2008} or filter design). Therefore, knowing how to exploit the X-channel scenario in the setup under consideration (which is in part characterized by the power
domain operation) in this paper, appears to a relevant extension ($\star$).

\textbf{Channel gain quantization operation $\mathcal{Q}_{i}^{\mathrm{II}}$}: The first step in Phase $\mathrm{II}$ is for each of the transmitters to quantize the $K-$dimensional vector $\widetilde{\ul{g}_i}$.  For simplicity, we assume that each element of the real $K-$dimensional vector $\widetilde{\ul{g}_i}$ is quantized by a scalar quantizer into a label of $N_{\mathrm{II}}$ bits. This assumption is motivated by low complexity but also by the fact that the components of $\widetilde{\ul{g}_i}$ are independent in the most relevant scenarios of interest. For instance, if local CSI is very well estimated, the estimated channel gains are close to the actual channel gains, which are typically independent in practice. Now, in the general case of arbitrary estimation noise level, the components of $\widetilde{\ul{g}_i}$ will be independent when the training matrix $\mathbf{P}_{\mathrm{I}}$ is chosen to be diagonal, which is a case of high interest and is motivated further in Sec. V. Under the channel gain (quasi-) independency, vector quantization would bring (almost) no performance improvement. The \textit{scalar quantizer} used by Transmitter $i$ to quantize $\widetilde{g}_{ji}$ is denoted by $\mathcal{Q}_{ji}^{\mathrm{II}}$. Finding the best quantizer in terms of ultimate network utility (e.g., in terms of sum-rate or sum-energy-efficiency) does not appear to be straightforward ($\star$). We present two possible quantization schemes in this section.

A possible, but generally sub-optimal approach, is to determine a quantizer which \textit{minimizes distortion}. The advantage of such approach is that it is possible to express the quantizer and it leads to a scheme which is independent of the network utility; this may be an advantage when the utility is unknown or changing. A possible choice for the quantizer $\mathcal{Q}_{i}^{\mathrm{II}}$ is to use the conventional version of the Lloyd-Max algorithm (LMA) \cite{lloyd-tit-1982}. However, this algorithm assumes perfect knowledge of the information source to be quantized (here this would amount to assuming the channel estimate to be noiseless) and no noise between the quantizer and the dequantizer (here this would amount to assuming perfect knowledge of the RS power). The authors of \cite{Brice07} proposed a generalized version of the Lloyd-Max algorithm for which noise can be present both at the source and the transmission but the various noise sources are assumed to verify standard assumptions (such as independence of the noise and the source), which are not verified in the setting under investigation; in particular, the noise in Phase $\mathrm{I}$ is the estimation noise, which is correlated  with the transmitted signal. Deriving the corresponding generalized Lloyd-Max algorithm can be checked to be a challenging task, which is left as an extension of the technical solutions proposed here ($\star$). Rather, we will provide here a special case of the generalized Lloyd-Max algorithm, which is very practical in terms of computational complexity and required knowledge. 

The version of the Lloyd-Max algorithm we propose will be referred to as \textit{ALMA (advanced Lloyd-Max algorithm)}. ALMA corresponds to the special case (of the most generalized version mentioned previously) in which the algorithm assumes noise on the transmission but not at the source (although the source can be effectively noisy). This setting is very well suited to scenarios where the estimation noise due to Phase I is negligible or when local CSI can be acquired reliably by some other mechanism. In the numerical part, we can observe the improvements of the proposed ALMA with respect to the conventional LMA. Just like the conventional LMA, ALMA aims at minimizing distortion by iteratively determining the best set of representatives and the best set of cells (which are intervals here) when one of the two is fixed. The calculations for obtaining the optimal representatives and partitions are given in Appendix C for both the special case of no source noise as well as for the general case. Solving the general case can be seen from Appendix C to be computationally challenging.

To comment on the proposed algorithm which is given by the pseudo-code of Algorithm 1, a few notations are in order. We denote by $q \in \{1,...,Q\}$ the \textit{iteration index} (where $Q$ is the upper bound on the number of iterations) and define $R =2^{N_{\mathrm{II}}}$. For each channel gain estimate $\widetilde{g}_{ji}$ to be quantized, we denote by $\underline{v}_{ji} =\left\{v_{ji,1}^{(q)},...,v_{ji,R}^{(q)}\right\}$ the set of \textit{representatives} and by $\left\{u_{ji,1}^{(q)},...,u_{ji,R+1}^{(q)}\right\}$ (with $u_{ji,1}^{(q)}=0$ and $u_{ji,R+1}^{(q)}=\infty$) the set of \textit{interval bounds} which defines how the set $\widetilde{g}_{ji}$ lies in (namely $[0,+\infty)$) is partitioned. At each iteration, the choice of the set of representatives or intervals aims at minimizing the end-to-end distortion $\mathbb{E}|\widetilde{g}_{ji}-g_{ji}|^2$. This minimization operation requires some statistical knowledge. Indeed, the probability that the dequantizer decodes the representative $v_{ji,r}^{(q)}$ given that $v_{ji,n}^{(q)}$ has been transmitted needs to be known; this probability is denoted by $\pi_{ji}(r|n)$ and constitutes one of the inputs of Algorithm 1. The second input of Algorithm 1 is the p.d.f. of $g_{ji}$ which is denoted by $\phi_{ji}$. The third input is given by the initial choice for the quantization intervals that is, the set $\left\{u_{ji,1}^{(0)},...,u_{ji,R+1}^{(0)}\right\}$. Convergence of ALMA to a global minimum point is not guaranteed and finding sufficient condition for global convergence is known to be non-trivial. However, local convergence is guaranteed; an elegant and general argument for this can be found in \cite{larrousse-icassp-2014}. Conducting a theoretical analysis in which global convergence is tackled would constitute a significant development of the present analysis ($\star$), which is here based on typical and realistic simulation scenarios.

\begin{algorithm}[h]\label{algo}
{\bf{Inputs:}} {$\pi_{ji}$}, $\phi_{ji}\left(g_{ji}\right)$, $\left\{u_{ji,1}^{(0)},...,u_{ji,R+1}^{(0)}\right\}$\\
{\bf{Outputs:}} $\left\{u_{ji,1}^{\star},...,u_{ji,R+1}^{\star}\right\}$, $\left\{v_{ji,1}^{\star},...,v_{ji,R+1}^{\star}\right\}$\\ 
{\bf{Initialization:}} Set $q=0$. Initialize the quantization intervals according to $\left\{u_{ji,1}^{(0)},...,u_{ji,R+1}^{(0)}\right\}$. Set $u_{ji,r}^{(-1)}=0$ for all $r \in \{1,...,R\}$.\\
\While{$ \displaystyle{\max_r}   ||u_{ji,r}^{(q)}-u_{ji,r}^{(q-1)}||>\delta$ and $q<Q$}{
Update the iteration index: $q \gets q+1$.\\
For all $r \in \{1,2,..,R\}$ set
\begin{equation}
v_{ji,r}^{(q)} \gets
\frac{{\displaystyle\sum_{n=1}^R}  \pi_{ji} \left({r}|n\right)\displaystyle\int_{  u_{ji,n}^{(q-1)}   }^{ u_{ji,n+1}^{(q-1)} }
{g}_{ji}
\phi_{ji}({g}_{ji})
\mathrm{d}g_{ji} 
}{{\displaystyle\sum_{n=1}^R} \pi_{ji}  \left({r}|n\right)\displaystyle
\int_{  u_{ji,n}^{(q-1)}   }^{ u_{ji,n+1}^{(q-1)} }
\phi_{ji}({g}_{ji})
\mathrm{d}g_{ji} 
}.
\label{eq:alm_r1}
\end{equation}
For all $r \in \{2,3,..,R\}$ set
\begin{equation}
u_{ji,r}^{(q)} \gets \frac{{\displaystyle\sum_{n=1}^R}\left[ \pi_{ji}  \left(n|{r}\right)-  \pi_{ji} \left(n|{r}-1\right) \right] \left(v_{ji,n}^{(q)}\right)^2  }{{\displaystyle2 \sum_{n=1}^R}\left[ \pi_{ji} \left(n|{r}\right)- \pi_{ji}\left(n|{r}-1\right)\right] v_{ji,n}^{(q)}    }.
\label{eq:alm_t1}
\end{equation}
}
$\forall r \in \{2,...,R\}, \> u_{ji,r}^\star=u_{ji,r}^{(q)}$, $u_{ji,1}^\star=0$ and $u_{ji,R+1}^\star=\infty$ \\
$\forall r \in \{1,...,R\}, \> v_{ji,r}^\star= v_{ji,r}^{(q)}$
\caption{\small Advanced Lloyd-Max algorithm (ALMA)}
\end{algorithm}

At this point two comments are in order. First, through (\ref{eq:alm_r1})-(\ref{eq:alm_t1}), it is seen that ALMA relies on some statistical knowledge which might not always be available in practice. This is especially the case for $\pi_{ji}$ and $\gamma_{ji}$ since the knowledge of channel distribution information (CDI, i.e., $\phi_{ji}$) is typically easier to be obtained. The CDI may be obtained by storing the estimates obtained during past transmissions and forming empirical means (possibly with a sliding window). If the CDI is time-varying, a procedure indicating to the terminals when to update the statistics might be required. Second, if we regard Phase II as a classical communication process, then the amount of information sent by the source is maximized when the source signal is uniformly distributed. It turns out minimizing the (end-to-end) distortion over Phase II does not involve this. Motivated by these two observations we provide here a second quantization scheme, which is simple but will be seen to perform quite well in the numerical part. We will refer to this quantization scheme as \textit{maximum entropy quantizer (MEQ).} For MEQ, the quantization interval bounds are fixed once and for all according to:
\begin{equation}\label{eq:MEQ1}
\forall r \in \{1,...,R\}, \forall (j,i) \in \{1,...,K\}^2, 
\ \int_{u_{ji,r}}^{  u_{ji,r+1} } \phi_{ji}(g_{ji}) \mathrm{d} g_{ji} 
= \frac{1}{R}.
\end{equation}
The representative of the interval $[{u_{ji,r}}, {  u_{ji,r+1} } ]$ is denoted by $v_{ji,r}$ and is chosen to be its centroid:
\begin{equation}\label{eq:MEQ2}
v_{ji,r} = \frac{\displaystyle{ \int_{ u_{ji,r}   }^{ u_{ji,r+1} }  g_{ji} \phi_{ji}(g_{ji} ) \mathrm{d}   g_{ji}  } }{ \displaystyle{  \int_{ u_{ji,r}   }^{ u_{ji,r+1} }   \phi_{ji}(g_{ji} ) \mathrm{d}   g_{ji} }   }.
\end{equation}
We see that each representative has the same probability to occur, which maximizes the entropy of the quantizer output, hence the proposed name.  To implement MEQ, only the knowledge of $\phi_{ji}$ is required. Additionally, the complexity involved is very low.

%

\textbf{Power modulation:} To inform the other transmitters about its knowledge of local CSI, Transmitter $i$ maps the $K$ labels of $N_{\mathrm{II}}$ bits produced by the quantizer $Q_i^{\mathrm{II}}$ to a sequence of power levels $(p_i(T_{\mathrm{I}}+1),p_i(T_{\mathrm{I}}+2),\dots, p_i(T_{\mathrm{I}}+T_{\mathrm{II}}) )$. Any one-to-one mapping might be used a priori. Although the new problem of finding the best mapping for a given network utility arises here and constitutes a relevant direction to explore ($\star$), we will not only develop this here. Rather, our main objective here is to introduce this problem and illustrate it clearly for a special case which is treated in the numerical part. To this end, assume Phase II comprises $T_{\mathrm{II}}=2$ time-slots, $K=2$ users, and that the users only exploit $L=2$ power levels during Phase II say $\mathcal{P}= \{P_{\min}, P_{\max} \}$. Further assume $1-$bit quantizers, which means that the quantizers $\mathcal{Q}_{ji}^{\mathrm{II}}$ produce binary labels. For simplicity, we assume the same quantizer $\mathcal{Q}$ is used for all the four channel gains $g_{11}$, $g_{12}$, $g_{21}$, and $g_{22}$: if $g_{ij} \in [0,\mu]$ then the quantizer output is denoted by $g_{\min}$; if $g_{ij} \in (\mu, +\infty)$ then the quantizer output is denoted by $g_{\max}$.  Therefore a simple mapping scheme for Transmitter $1$ (whose objective is to inform Transmitter $2$ about $(g_{11}, g_{21})$) is to choose $p_1(T_{\mathrm{I}}+1) = P_{\min}$ if $\mathcal{Q}(g_{11})  = g_{\min} $ and $p_1(T_{\mathrm{I}}+1)  = P_{\max}$ otherwise; and $p_1(T_{\mathrm{I}}+2)  = P_{\min}$ if $\mathcal{Q}(g_{21})  = g_{\min} $ and $p_1(T_{\mathrm{I}}+2)  = P_{\max}$ otherwise. Therefore, depending on the p.d.f. of $g_{ij}$, the value of $\mu$, the performance criterion under consideration, a proper mapping can chosen. For example, to minimize the energy consumed at the transmitter, using the minimum transmit power level $P_{\min}$ as much as possible is preferable; thus if $\Pr({\mathcal{Q}(g_{11})  = g_{\min}})\geq \Pr({\mathcal{Q}(g_{11})  = g_{\max}})$, the power level $P_{\min}$ will be associated with the minimum quantized channel gain that is $\mathcal{Q}(g_{11})  = g_{\min}$. 

\textbf{Power level decoding: }For every time-slot $t \in \{T_{\mathrm{I}}+1, ..., T_{\mathrm{I}}+T_{\mathrm{II}}\}$ the power levels are estimated by Transmitter $i$ as follows
\begin{equation}\label{eq:lattice-decoding}
\ul{\widetilde{p}}_{-i}(t) \in \displaystyle{\arg \min_{\ul{p}_{-i} \in \mathcal{P}^{K-1}}   \left| \sum_{j \neq i}p_j \widetilde{g}_{ji} - (\widetilde{\omega}_i(t) - p_i(t)  \widetilde{g}_{ii}  - \sigma^2 )      \right|},
\end{equation} 
where $\ul{p}_{-i}=\left(p_1,..,p_{i-1},
p_{i+1},..,p_{K}\right)$. As for every $j$, $\widetilde{g}_{ji}$ is known at Transmitter $i$, the above minimization operation can be performed. It is seen that exhaustive search can be performed as long as the number of tests, which is $L^{K-1} $, is reasonable. For this purpose, one possible approach is to impose the number of power levels which are exploited over Phase $\mathrm{II}$  to be small. In this respect, using binary power over Phase $\mathrm{II}$ is not only relevant regarding complexity issues but also in terms of robustness against the various possible sources of noise. As for \textbf{the number of interfering users} using the same channel (meaning operating on the same frequency band, at the same period of time, in the same geographical area), it \textbf{will typically be small and does not exceed $3$ or $4$ in real wireless systems}. More generally, this shows that the proposed technique can accommodate more than $4$ users in total; For example, if we have $12$ bands, having $48=12\times 4$ users would be manageable by applying the proposed technique for each band. As our numerical results indicate, using (\ref{eq:lattice-decoding}) as a decoding rule to find the power levels of the other transmitters generally works very well for $K=2$. When the number of users is higher, each transmitter  needs to estimate $K-1$ power levels  with only one observation equation, which typically induces a non-negligible degradation in terms of symbol error rate. In this situation, Phase II can be performed by scheduling the activity of all the users, such that only $2$ users are active at any given time-slot in Phase II. Once all pairs of users have exchanged information on their channel states, Phase II is concluded. 

\textit{Remark 1.} Note that the case where only one user is active at a time is a special case of the decoding scheme assumed here. The advantage of our more general decoding scheme is that it can be used when strict SINR feedback is used \cite{varma-eusipco-2015} instead of RSSI; indeed when only one user is active at a time, the SINR becomes an instantaneous SNR and cannot convey any coordination information. Concerning the setting with RSSI feedback, the drawback of our assumption is that in the presence of noise on the RS power feedback, the performance of Phase II may be limited when the cross channel gains are very small. If this turned out to be a crucial problem, allowing only one user to be active at a time is preferable.


\textit{Remark 2 (required number of time-slots).} The proposed technique typically requires $K+K=2K$ time-slots for the whole exploration phase (Phases I and II). It therefore roughly require the same amount of resources as IWFA, which indeed needs about $2K$ or $3K$ SINR samples to converge to Nash equilibrium. While channel acquisition may seem to take some time, please note that regular communication is uninterrupted and occurs in parallel. {As already mentioned, the context in which the proposed technique and IWFA are the most suited is a context where the channel is constant over a large number of time-slots, which means that the influence of the exploration phase on the average performance is typically negligible. Nonetheless, some simulations will be provided to assess the optimality loss induced by using power levels to convey information.}

\textit{Remark 3 (extension to the multi-band scenario).}  As explained in the beginning of this paper, Phases I and II are described for the single-band case, mainly for clarity reasons. Here, we briefly explain how to adapt the algorithm when there are multiple bands. In Phase I, the only difference exists in choosing the training matrix. With say $S$ bands to transmit, for each band $s \in \{1,...,S\}$, the training matrix $\mathbf{P}_{\mathrm{I}}^{s}$ has to fulfill the constraint $\sum_{s=1}^Sp_i^s(t) \leq P_{\max}$ where $p_i^s$ is the power Transmitter $i$ allocates to band $s$. In Phase II, each band performs in parallel like the single-band case. Since there are power constraints for each transmitter, the modulated power should satisfy  $\sum_{s=1}^Sp_i^s(t) \leq P_{\max}$. 

\textit{Remark 4 (extension to the multi-antenna case).} To perform operation such as beam-forming, the phase information is generally required. The proposed local CSI estimation techniques (namely, for Phase I) do not allow the phase information or the direction information to be recovered;  Therefore, another type of feedback should be considered for this. However, if another estimation scheme is available or used for local CSI acquisition and that scheme provides the information phase, then the techniques proposed for local CSI exchange (namely, for Phase II) can be extended. An extension which is more in line with the spirit of the manuscript is given by a MIMO interference channel for which each transmitter knows the interference-plus-noise covariance matrix and its own channel. This is the setup assumed by Scutari et al in their work on MIMO iterative water-filling \cite{scutari-tsp-2009}.

{\textit{Remark 5 (type of information exchanged).} One of the strengths of the proposed exchange procedure is that \textbf{any kind} of information can be exchanged. However, since SINR or RSSI is used as the communication channel, this has to be at a low-rate which is given by the frequency at which the power control levels are updated and the feedback samples sent.}
\section{Numerical Analysis}
\label{sec:numerical-analysis}

In this section, as a first step (Sec.~\ref{sec:global-perf-analysis}), we start with providing simulations which result from the combined effects of Phases I and II. To make a coherent comparison with IWFA, the network utility will be evaluated without taking into account a cost possibly associated with the exploration or training phases (i.e., Phases I and II for the proposed scheme or the convergence time for IWFA). The results are provided for a reasonable scenario of small cell networks which is similar to those already studied in other works (see e.g., \cite{samarakoon-jsac-sub-2016} for a recent work). As a second step (Sec.~\ref{sec:phase-I} and \ref{sec:phase-II}), we study special cases to better understand the influence of each estimation phase and the different parameters which impact the system performance.


\subsection{Global performance analysis: a simple small cell network scenario}
\label{sec:global-perf-analysis}

\begin{minipage}{\linewidth}
\centering
\begin{minipage}{0.4 \linewidth}
\begin{table}[H]
\begin{tabular}{|p{1.8cm}|p{4cm}|p{1.2cm}|}
\hline
\textbf{Acronym} & \textbf{Meaning} & \textbf{Definition} \\
\hline
ALMA & advanced Lloyd-Max algorithm & (\ref{eq:alm_r1}),(\ref{eq:alm_t1}) \\
\hline
CSI & channel state information & \\
\hline
EE & energy-efficiency & \\
\hline
ENSR & estimation signal-to-noise ratio & (\ref{eq:ESNR})\\
\hline
ISD & inter site distance &\\
\hline
IWFA & iterative water-filling algorithm & \cite{yu-jsac-2002} \\
\hline
LMA & conventional Lloyd-Max algorithm & \cite{lloyd-tit-1982} \\
\hline
LSPD & least squares estimator  &  (\ref{eq:LS}) \\
& in power domain & \\
\hline
MEQ & maximum entropy quantizer &  (\ref{eq:MEQ1}),(\ref{eq:MEQ2}) \\
\hline
MMSEPD & minimum mean square error  & (\ref{eq:MMSEPD}) \\
& estimator in power domain & \\
\hline
MS & mobile station & \\
\hline
SBS & small base station & \\
\hline
Team BRD & team best response dynamics & (\ref{eq:sum-rate-def}),(\ref{eq:sum-EE-def})\\
\hline
\end{tabular}
\vspace{-.01cm}
\caption{\small Acronyms used in Sec.~\ref{sec:numerical-analysis}}
\end{table}
\end{minipage}
\hspace{0.1\linewidth}
\begin{minipage}{0.45 \linewidth}
\begin{figure}[H]\label{fig:small-cells-network}
\begin{center}
\large
\begin{tikzpicture}[
    xscale = .52,yscale=.52, transform shape, thick,
    grow = down,  
    level 1/.style = {sibling distance=3cm},
    level 2/.style = {sibling distance=4cm}, 
    level 3/.style = {sibling distance=2cm}, 
    level distance = 1.25cm
  ]
 \draw [dotted] (10,0) -- (10,15);
\draw [dotted] (0,5) -- (15,5);
\draw [dotted] (0,10) -- (15,10); 
\draw (0,0) -- (15,0) -- (15,15) -- (0,15) -- (0,0);
\draw [dotted] (5,0) -- (5,15);

\node (b1) at (2.5,2.5)[shape = circle,draw]{SBS$_1$};
\node (b2) at (7.5,2.5)[shape = circle,draw]{SBS$_2$};
\node (b3) at (12.5,2.5)[shape = circle,draw]{SBS$_3$};
\node (b4) at (2.5,7.5)[shape = circle,draw]{SBS$_4$};
\node (b5) at (7.5,7.5)[shape = circle,draw]{SBS$_5$};
\node (b6) at (12.5,7.5)[shape = circle,draw]{SBS$_6$};
\node (b7) at (2.5,12.5)[shape = circle,draw]{SBS$_7$};
\node (b8) at (7.5,12.5)[shape = circle,draw]{SBS$_8$};
\node (b9) at (12.5,12.5)[shape = circle,draw]{SBS$_9$};
      

\node (m1) at (3.91101, 3.6508) [shape = rectangle,draw,fill=gray!10]{MS$_1$};
\node (m2) at (7.9828 ,1.01402) [shape = rectangle,draw,fill=gray!10]{MS$_2$};
\node (m3) at (10.1921 ,0.7040) [shape = rectangle,draw,fill=gray!10]{MS$_3$};
\node(m4) at (2.2934 ,5.8828) [shape = rectangle,draw,fill=gray!10]{MS$_4$};
\node(m5) at (6.6272 ,5.9110) [shape = rectangle,draw,fill=gray!10]{MS$_5$};
\node(m6) at (14.1351 ,9.2577) [shape = rectangle,draw,fill=gray!10]{MS$_6$};
\node(m7) at (1.8307 ,10.5872) [shape = rectangle,draw,fill=gray!10]{MS$_7$};
\node(m8) at (7.1415 ,14.6032) [shape = rectangle,draw,fill=gray!10]{MS$_8$};
\node(m9) at (12.5028 ,10.6665) [shape = rectangle,draw,fill=gray!10]{MS$_9$};

  \draw [->] (b9) to node[rectangle,left]{} (m9);
  \draw [->] (b1) to node[rectangle,left]{} (m1);
  \draw [->] (b2) to node[rectangle,left]{} (m2);
  \draw [->] (b3) to node[rectangle,left]{} (m3);
  \draw [->] (b4) to node[rectangle,left]{} (m4);
  \draw [->] (b5) to node[rectangle,left]{} (m5);
  \draw [->] (b6) to node[rectangle,left]{} (m6);
  \draw [->] (b7) to node[rectangle,left]{} (m7);
  \draw [->] (b8) to node[rectangle,left]{} (m8);

\draw [<->,color=blue,double] (10.2,14.8) to node[rectangle,fill=gray!1,below,yshift=-.3cm,xshift=.7cm]{Cell size: $d\times d$} (14.5,14.8);

\draw [<->,color=blue,double] (14.8,10.2) to node[rectangle]{} (14.8,14.5);

\draw [<->,double,color=blue] (2.5,11.5) to node[rectangle,fill=gray!1,below,yshift=-.3cm]{Inter-site distance: $d$} (7.5,11.5);
\draw [<->,double,color=blue] (8.5,12.5) to node[rectangle,fill=gray!1,right,xshift=.1cm]{$d$} (8.5,7.5);

%
  \draw [->,dashed,color=red] (b7) to node[rectangle,above,yshift=.3cm,xshift=-.3cm,fill=gray!1]{Interference} (m8);

%
%
%
%

\end{tikzpicture}
\caption{\small Small cell network configuration assumed in Sec.~\ref{sec:global-perf-analysis}}
\end{center}

\end{figure}
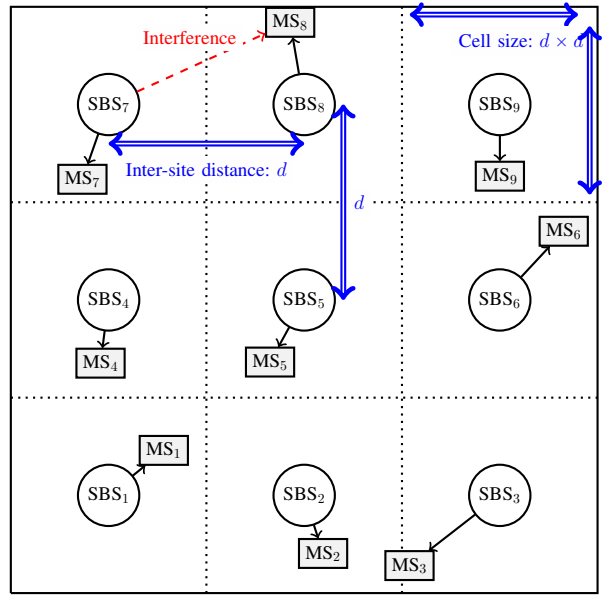
\end{minipage}
   \end{minipage}

\vspace{0.5cm}
As shown in Fig.~3, the considered scenario assumes $K=9$ small cell base stations with maximal transmit power $P_{\max} = 30$ dBm. One or two bands are assumed, depending on the scenario considered. One user per cell is assumed, which corresponds to a possible scenario in practice (see e.g., \cite{samarakoon-jsac-sub-2016}\cite{moustakas-isit-2013}\cite{moustakas-tit-sub-2013}). We also use this setup to be able to compare the proposed scheme with IWFA whose performance is generally assessed for the most conventional form of the interference channel, namely, $K$ transmitter-receiver pairs. The normalized receive noise power is $\sigma^2 = 0$ dBm. This corresponds to $\mathrm{SNR(dB)}=30$ where the \textit{signal-to-noise ratio} is defined by
\begin{equation}
\mathrm{SNR(dB)}=10 \log_{10}\left(  \frac{P_{\max}}{\sigma^2} \right). 
\end{equation}
Here and in all the simulation section, we set the SNR to $30 \ \mathrm{dB}$ by default. RS power measurements are quantized uniformly in a dB scale with $N=8$ bits and the quantizer input dynamics or range in dB is $[\mathrm{SNR(dB)}-20, \mathrm{SNR(dB)}+ 10]$. The DMC $\Gamma$ is constructed with error probability  $\epsilon$ to the two nearest neighbors, i.e., for the symbols $\mathrm{w}_1<\mathrm{w}_2<\dots<\mathrm{w}_M$ (with $M=2^N$), $\Gamma\left(\mathrm{w}_{i}|\mathrm{w}_{j}\right)=\epsilon$ if $|{i-j}| = 1$ and $\Gamma\left(\mathrm{w}_{i}|\mathrm{w}_{j}\right)=0$ if $|{i-j}| > 1$. In this section $\epsilon = 1\%$; the quantity $\epsilon$ will be referred to as the \textit{feedback channel symbol error rate} (FCSER). For all $(i,j)$ and $s$ ($s$ always being the band index) the channel gain $g_{ij}^s$ on band $s$ is assumed to be exponentially distributed namely, its p.d.f. writes as $\phi_{ij}^s(g_{ij}^s) = \frac{1}{\mathbb{E}[g_{ij}^s ]}  \exp\left(-\frac{g_{ij}^s}{\mathbb{E}[g_{ij}^s]} \right) $; this corresponds to the well-known Rayleigh fading assumption. Here, $\mathbb{E}(g_{ij}^s)$ models the path loss effects for the link $ij$ and depends of the distance as follows: $\mathbb{E}(g_{ij}^s) = \left(\frac{d_0}{d_{ij}}\right)^2$ where $d_{ij}$ is the distance between Transmitter $i$ and Receiver $j$ and $d_0= 5$ m is a normalization factor. The normalized coordinates of the mobile stations $\mathrm{MS}_1, ..., \mathrm{MS}_9$ are respectively given by: $(3.8, 3.2)$, $(7.9, 1.4)$, $(10.2, 0.7)$, $(2.3, 5.9)$, $(6.6, 5.9)$, $(14.1, 9.3)$, $(1.8, 10.6)$, $(7.1, 14.6)$, $(12.5, 10.7)$; the real coordinates are obtained by multiplying the former by the ratio $\frac{\mathrm{ISD}}{d_0}$, ISD being the inter site distance. In this section, the system performance is assessed in terms of sum-rate, the \textit{sum-rate} being given by:
\begin{equation}\label{eq:sum-rate-def}
u^{\text{sum-rate}}(\ul{p}_1,...,\ul{p}_K; \Gm) = \sum_{i=1}^K \sum_{s=1}^S\log( 1+ \mathrm{SINR}_i^s(  \ul{p}_1,...,\ul{p}_K ;\Gm)).
\end{equation}
where $\ul{p}_i = (p_{i}^1,..., p_i^S)$ represents the power allocation vector of Transmitter $i$, $\mathrm{SINR}_i^s$ is the SINR at Receiver $i$ in band $s$ and expresses as $\mathrm{SINR}_{i}^s = \frac{ g_{ii}^s  p_{i}^s}{\displaystyle{ \sigma^2+\sum_{j \neq i} g_{ji}^s  p_{j}^s }}$.

Fig.~\ref{fig:global-sum-rate-ISD-S2}, represents the average sum-rate against the ISD. The sum-rate is averaged over $10^4$ realizations of the channel gain matrix $\Gm$ and the \textit{inter site distance} is the distance between two neighboring small base stations. Three curves are represented. The top curve corresponds to the performance of the sequential best-response dynamics applied to the sum-rate (referred to as \textit{Team BRD}) in the presence of perfect global CSI. The curve in the middle corresponds to Team BRD which uses the estimate obtained by using the most simple association proposed in the paper namely, LSPD for Phase I and the $2-$bit MEQ for Phase II. The LSPD estimator uses $K$ time-slots and the $K-$dimensional identity matrix $\mathbf{P}_{\mathrm{I}} = P_{\max} \mathbf{I}_K$ for the training matrix. The $2-$bit MEQ uses binary power control ($L=2$) and $2K$ time-slots to send the information, i.e.,  $\ul{g}_i$); this corresponds to the typical number of time-slots IWFA needs to converge. At last, the bottom curve corresponds to IWFA using local CSI estimates provided by Phase I. It is seen that about $50\%$ of the gap between IWFA and Team BRD with perfect CSI can be bridged by using the proposed estimation procedure. When the interference level is higher, the gap becomes larger. Fig.~\ref{fig:global-sum-rate-ISD-S1} depicts exactly the same scenario as Fig.~\ref{fig:global-sum-rate-ISD-S2} except that only one band is available to the small cells i.e., $S=1$. Here the gap can be bridged at about $65\%$ when using Team BRD with the proposed estimation procedure.

In this section, some choices have been made:  a diagonal training matrix and the LSPD estimator has been chosen for Phase I and the MEQ has been chosen for Phase II. The purpose of the next sections is to explain these choices, and to better identify the strengths and weaknesses of the proposed estimation procedures.  
\vspace{-0.15cm}
\begin{figure}[h]
	\centering
	\begin{subfigure}[t]{.48 \linewidth}
		\centering
		\includegraphics[width=1.0 \linewidth]{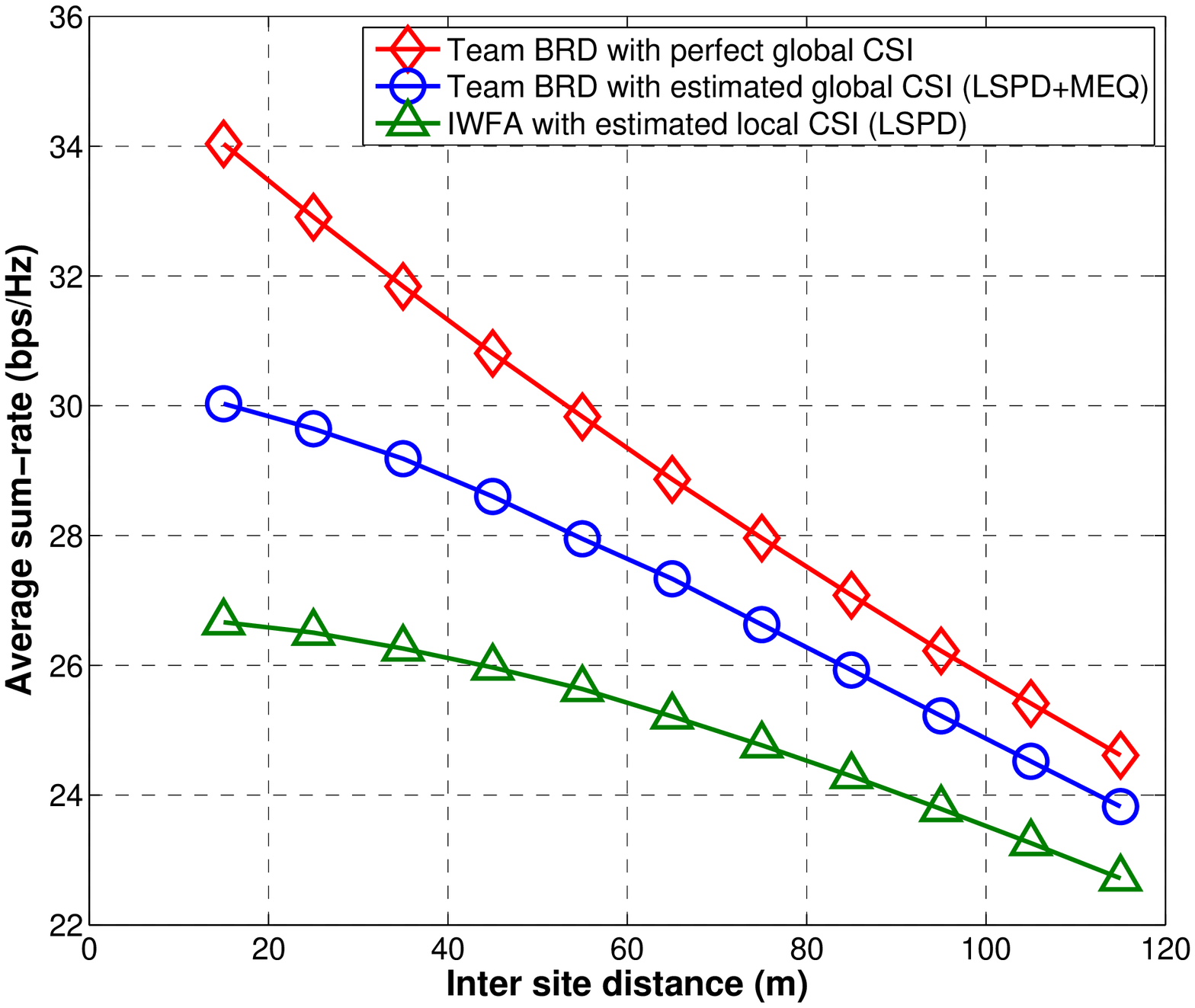}
		\caption{$S=2$}\label{fig:global-sum-rate-ISD-S2}
	\end{subfigure}
\hspace{0.01\linewidth}
	\begin{subfigure}[t]{.48 \linewidth}
		\centering
		\includegraphics[width=1.0 \linewidth]{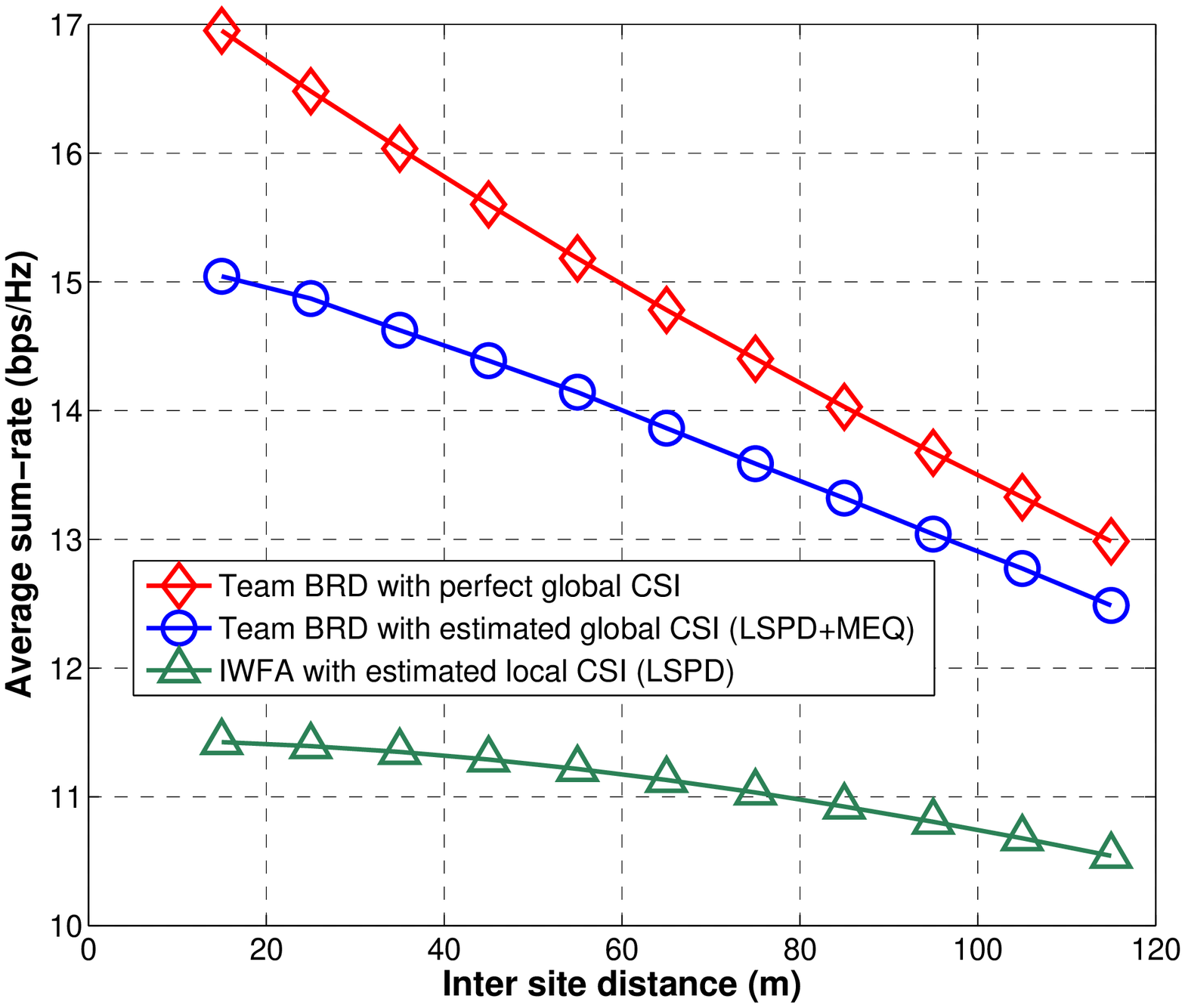}
		\caption{$S=1$}\label{fig:global-sum-rate-ISD-S1}
	\end{subfigure}
	\caption{\small The above curves are obtained in the scenario of Fig.~4 in which $K=9$ transmitter-receiver pairs, $\mathrm{SNR(dB)=30}$, the FCSER is given by $\epsilon=0.01$, $N=8$ quantization bits for the RSSI, and $L=2$ power levels. Using the most simple estimation schemes proposed in this paper namely LSPD and MEQ can bridge the gap between the IWFA and the team BRD with perfect CSI, about $50\%$ when $S=2$ and about $65\%$ when $S=1$.}
\end{figure}
\vspace{-0.15cm}
\vspace{-0.15cm}
\subsection{Comparison of estimation techniques for Phase I}
\label{sec:phase-I}

\begin{figure}[h]
\raggedleft
	\begin{subfigure}[t]{.48 \linewidth}
		\centering
		\includegraphics[width=1.0 \linewidth]{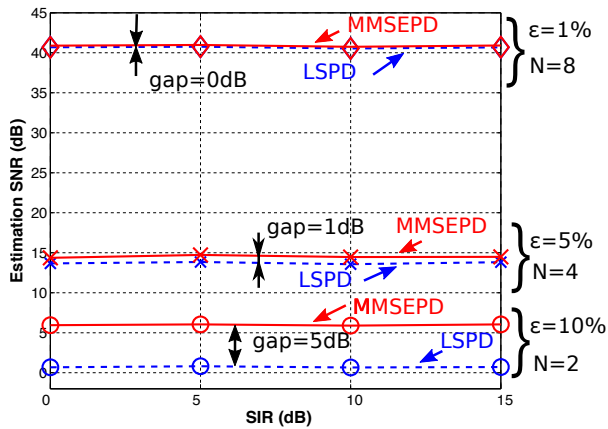}
		\caption{Using MMSEPD instead of LSPD in Phase I becomes useful in terms of ESNR when the RSSI quality becomes too rough (bottom curves).}\label{fig:phase-I-ESNR-SIR}
	\end{subfigure}
\hspace{0.01\linewidth}
	\begin{subfigure}[t]{.48 \linewidth}
\raggedright
		\includegraphics[width=1.0 \linewidth]{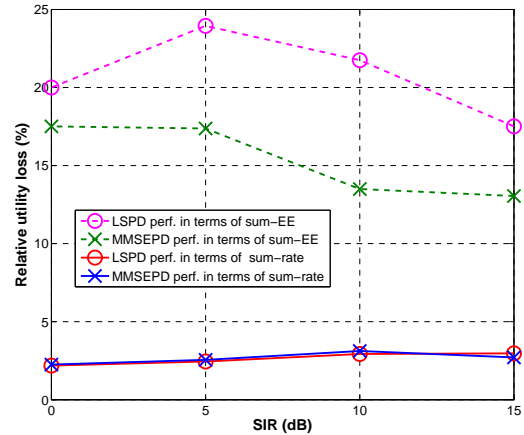}
		\caption{The figure provides the relative utility loss under quite severe conditions in terms of RSSI quality ($N=2$, $\epsilon=10\%$).}\label{fig:phase-I-Delta-u-SIR}
	\end{subfigure}
	\caption{\small Comparing MMSEPD and LSPD assuming perfect Phase II.}
\end{figure}

In Phase I, there are two main issues to be addressed: the choice of the estimator and the choice of the training matrix $\mathbf{P}_{\mathrm{I}}$. To compare the LSPD and MMSEPD estimators, we first consider the estimation SNR (ESNR) as the performance criterion to compare them. The \textit{estimation SNR} of Transmitter $i$ is defined here for the case $S=1$ and is given by:
\begin{equation}\label{eq:ESNR}
\mathrm{ESNR}_i= \frac{\mathbb{E}[\|\mathbf{G}\|^2]}{\mathbb{E}[\|\mathbf{G}-\widetilde{\mathbf{G}}_i\|^2]} .
\end{equation}
where $\|.\|^2$ stands for the Frobenius norm and $\widetilde{\mathbf{G}}_i$ is the global channel estimate which is available to Transmitter $i$ after Phases I and II. In this section, we always assume a perfect exchange in Phase II to conduct the different comparisons. This choice is made to isolate the impact of Phase I estimation techniques on the estimation SNR and the utility functions which are considered for the exploitation phase. After extensive simulations, we have observed that the gain in terms of ESIR by using the best training matrix (computed by an exhaustive search over all the matrix elements) is found to be either negligible or quite small when compared to the best diagonal training matrix (computed by an exhaustive search over the diagonal elements); see e.g., Fig.~6a for such a simulation. Therefore, for the rest of this paper, we will restrict our attention to diagonal training matrices for reducing the computational complexity without any significant performance loss. {To conclude about the choice of the training matrix, we assess the impact of using power levels to learn local CSI instead of using them to optimize the performance of Phase I. For this, we compare in Fig. 6b the scenario in which a diagonal training matrix is used to learn local CSI, with the scenario in which the best training matrix in the sense of the expected sum-rate (over Phase I). Global channel distribution information is assumed to be available in the latter scenario. The corresponding choice is feasible computationally speaking for small systems. }

\begin{figure}[h]
\raggedleft
	\begin{subfigure}[t]{.48 \linewidth}
		\centering
		\includegraphics[width=1.0 \linewidth]{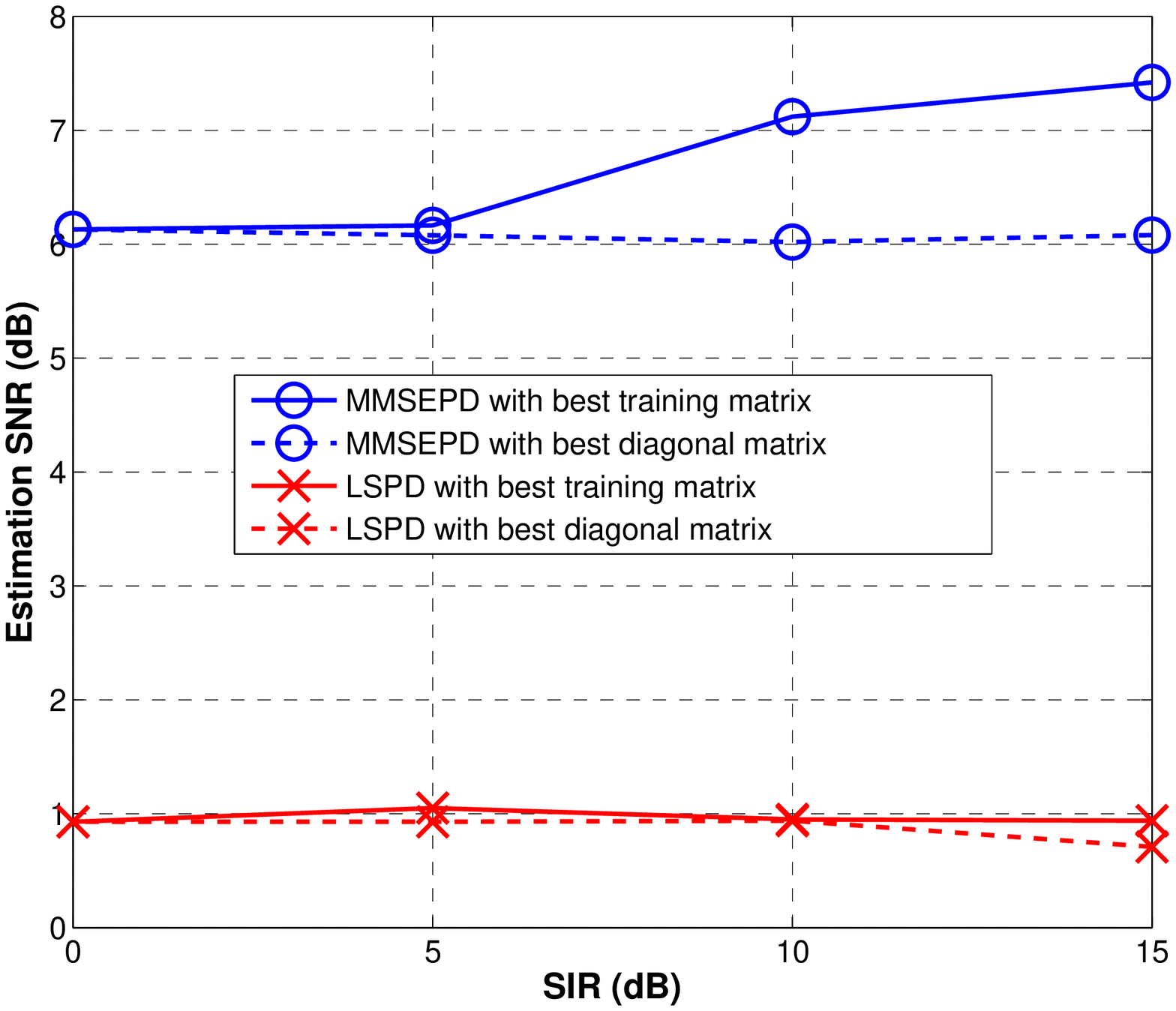}
		\caption{Scenario: $K=2$, $S=1$, and $\mathrm{SNR(dB)}=30$, $\epsilon=0$, $N=2$ quantization bits. Using a diagonal training matrix typically induces a small performance loss in terms of ESNR even in worst-case scenarios.}\label{fig:9}
	\end{subfigure}
\hspace{0.01\linewidth}
	\begin{subfigure}[t]{.48 \linewidth}
\raggedright
		\includegraphics[width=1.0 \linewidth]{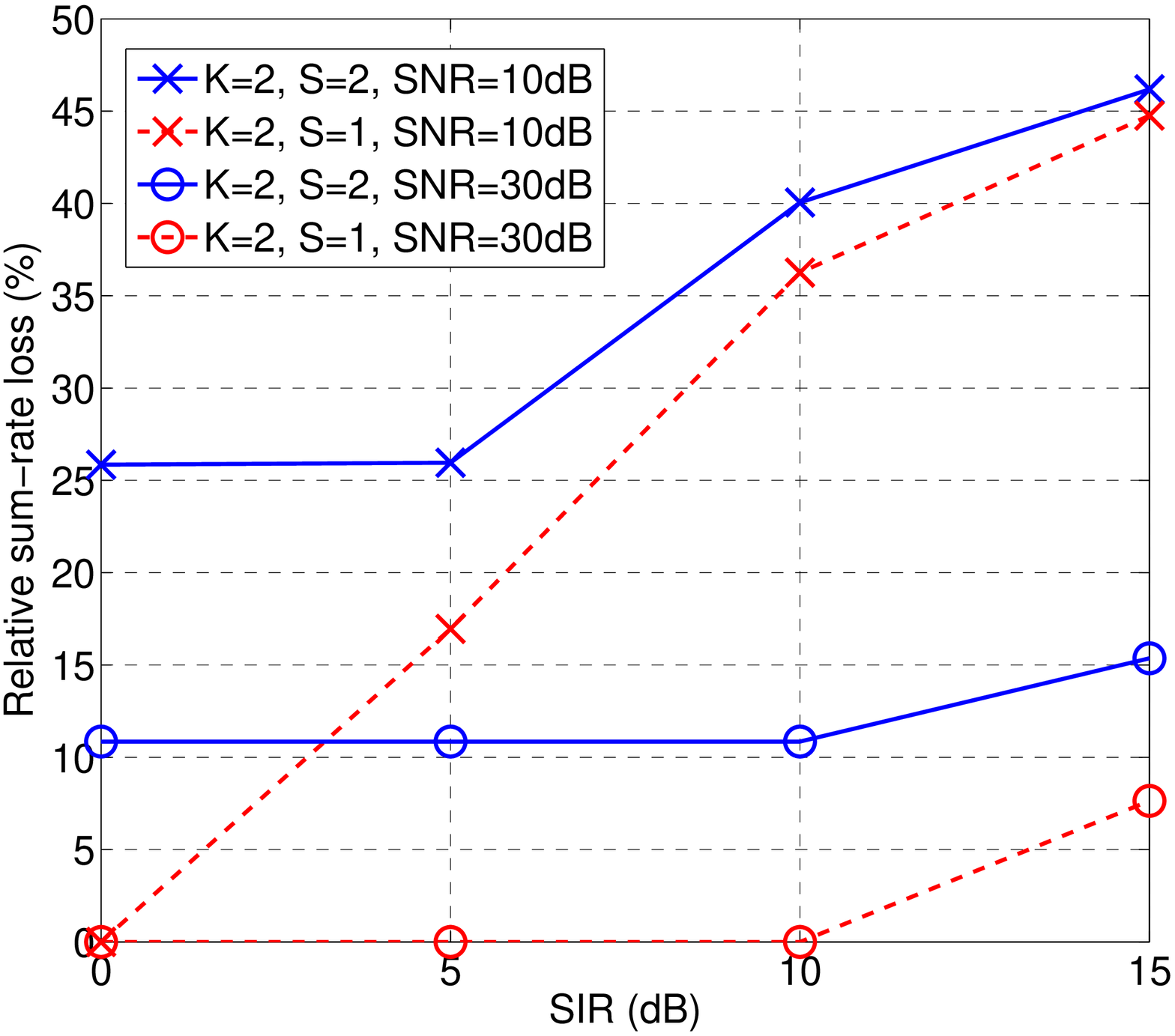}
		\caption{{Optimality loss induced in Phase I when using power levels to learn local CSI instead of maximizing the expected sum-rate. This loss may be influential on the average performance when the number of time-slots of the exploitation phase is not large enough.}}
	\end{subfigure}
	\caption{\small {Influence of the training matrix.}}
\end{figure}

Fig.~\ref{fig:phase-I-ESNR-SIR} represents for $K=2$, $S=1$, and $\mathrm{SNR(dB)}=30$, the estimation SNR (in dB) against the signal-to-interference ratio (SIR) in dB $\mathrm{SIR(dB)}$ which is defined here as 
\begin{equation}
\mathrm{SIR(dB)} = 10 \log_{10}\left(  \frac{\mathbb{E}(g_{11})}{\mathbb{E}(g_{21})} \right) = 10 \log_{10}\left(  \frac{\mathbb{E}(g_{22})}{\mathbb{E}(g_{12})} \right).  
\end{equation}
The three curves in red solid lines represent the MMSEPD estimator performance while the three curves in blue dashed line represent the LSPD estimator performance. The performance gap between MMSEPD and LSPD depends on the quality of the RSSI at the transmitters. When RS power measurements are quantized with $N=8$ bits and the feedback channel symbol error rate is $\epsilon = 1 \%$, the gap in dB is very close to $0$. Using MMSEPD instead of LSPD becomes much more relevant in terms of ESNR when the quality of feedback is degraded. Indeed, for $N=2$ bits and  $\epsilon = 10 \%$, the gap is about $5$ dB. Note that having a very small number of RSSI quantization bits and therefore significant feedback quality degradation may also occur in classical wireless systems where the feedback would be binary such as an ACK/NACK feedback. Indeed, an ACK/NACK feedback can be seen as the result of a $1-$bit quantization of the RSSI or SINR. The proposed technique might be used to coordinate the transmitters just based on this particular and rough feedback. Even though the noise on the RSSI is correlated with the signal and is not Gaussian, we observe that MMSEPD and LSPD (which can be seen as a zero-forcing solution) perform similarly when the noise becomes negligible. At last note that the $\mathrm{ESNR}$ is seen to be independent of the $\mathrm{SIR}$; this can be explained by the used training matrix, which is diagonal. 

The above comparison is conducted in terms of ESNR but not in terms of final utility. To assess the impact of Phase I on the exploration phase, two common utility functions are considered namely, the sum-rate and the \textit{sum-energy-efficiency} (sum-EE) which is defined as: 
\begin{equation}\label{eq:sum-EE-def}
u^{\text{sum-EE}}(\ul{p}_1,...,\ul{p}_K;\Gm) = \sum_{i=1}^K\frac{ 
\displaystyle{\sum_{s=1}^S} f(\mathrm{SINR}_i^s(\ul{p}_1,...,\ul{p}_K;\Gm))}{ \displaystyle{\sum_{s=1}^S}  p_{i}^s} .
\end{equation}
where the same notations as in (\ref{eq:sum-rate-def}) are used; $f$ is an efficiency function which represents the packet success rate or the probability of having no outage.  Indeed, the utility function $u^{\text{sum-EE}}$ corresponds to the ratio of the packet success rate to the consumed transmit power and has been used in many papers (see e.g., \cite{meshkati-jsac-2006}\cite{belmega-tsp-2011}\cite{buzzi-jstsp-2012}\cite{bacci-tsp-2013}\cite{haddad-sigmetrics-2014}). Here we choose the efficiency function of \cite{belmega-tsp-2011}: $f(x) = \exp\left(  - \frac{c}{x}\right)$ with $c=2^r - 1 =1$, $r$ being the spectral efficiency. Fig.~\ref{fig:phase-I-Delta-u-SIR} depicts for $K=2$, $S=1$, $N=2$, $\epsilon= 10 \%$ the \textit{average relative utility loss} $\Delta u$ in $\%$ against the SIR in dB. The average relative utility loss in $\%$ is defined by
\begin{equation}
\Delta u(\%) = 100 \mathbb{E} \left[ \frac{u(\ul{p}_1^\star,...,\ul{p}_K^\star; \Gm) - u(\widetilde{\ul{p}}_1^\star,...,\widetilde{\ul{p}}_K^\star; \Gm)}{u(\ul{p}_1^\star,...,\ul{p}_K^\star; \Gm)} \right].
\end{equation}
where $u(\ul{p}_1^\star,...,\ul{p}_K^\star; \Gm)$ is the best sum-utility which can be attained when every realization of $\Gm$ is \begin{wrapfigure}{r}{0.5\textwidth}
   \begin{center}
\includegraphics[width=0.48\textwidth]{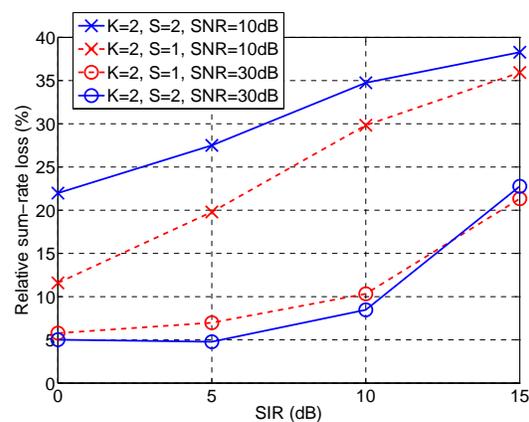}
  \end{center}
  \vspace{-0.55cm}
   \caption{\small {Optimality loss induced in Phase II when using power levels to exchange local CSI instead of maximizing the expected sum-rate. This loss may be influential on the average performance when the number of time-slots of the exploitation phase is not large enough.}}
    \label{fig:111}
\end{wrapfigure} known perfectly. The latter is obtained by performing exhaustive search over $100$ values equally spaced in $[0, P_{\max}]$ and this for each draw of $\Gm$; the average is obtained from $10^4$ independent draws of $\Gm$. The utility $u(\widetilde{\ul{p}}_1^\star,...,\widetilde{\ul{p}}_K^\star; \Gm)$ is also obtained with exhaustive search but by using either the LSPD or MMSEPD estimator and assuming Phase II to be perfect. Fig.~\ref{fig:phase-I-Delta-u-SIR} shows that even under severe conditions in terms observing the RS power at the transmitter, the MMSEPD and LSPD estimators have the same performance in terms of sum-rate.  This holds even though the gap in terms of ESNR is $5$ dB (see Fig.~\ref{fig:phase-I-ESNR-SIR}). Note that the relative utility loss is about $3\%$ showing that the sum-rate performance criterion is very robust against channel estimation errors. When one considers the sum-EE, the relative utility loss becomes higher and is the range $15\%-20\%$ and the gap between MMSEPD and LSPD becomes more apparent this time and equals about $5\%$. The observations made for the special setting considered here have been checked to be quite general and apply for more users, more bands, and other propagation scenarios: unless the RSSI is very noisy or when only an ACK/NACK-type feedback is available, the MMSEPD and LSPD estimators perform quite similarly. Since the MMSEPD estimator requires more knowledge and more computational complexity to be implemented, the LSPD estimator seems to be the best choice when the quality of RSSI is good as it is in current cellular and Wifi systems.

{To conclude this section, we provide the counterpart of Fig. 6b for phase 2 in Fig. 7. The scenario in which a diagonal training matrix is used to exchange local CSI, with the scenario in which power control is to maximize the expected sum-rate (over Phase II). But here, the expectation is not taken over local CSI since it is assumed to be known. The corresponding choice is feasible computationally speaking for small systems.}

%
\subsection{Comparison of quantization techniques for Phase II}
\label{sec:phase-II}

\begin{figure}[h]
	\centering
	\begin{subfigure}[t]{.48 \linewidth}
		\centering
		\includegraphics[width=1.0 \linewidth]{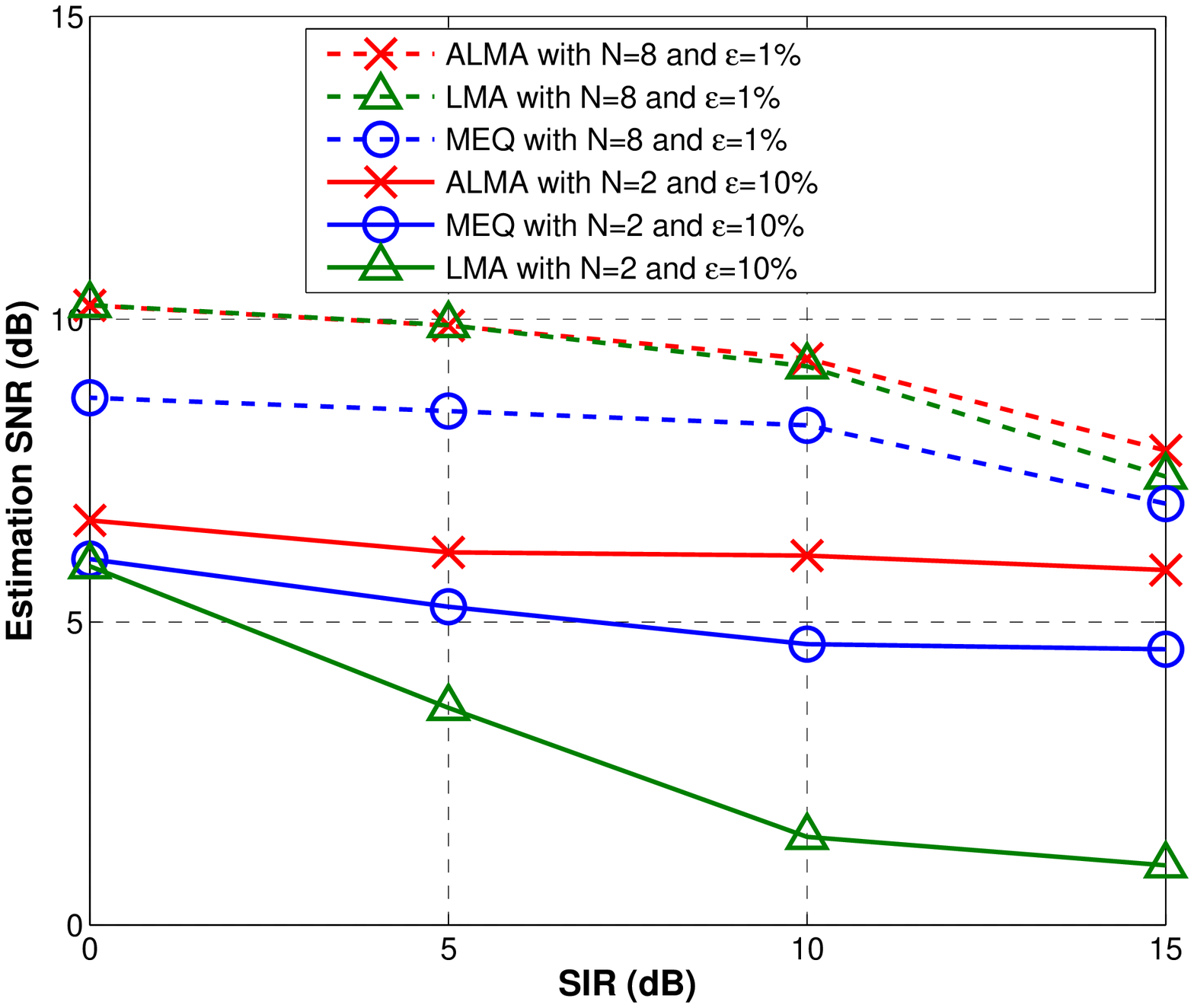}
		\caption{Performance measured by ESNR considering good (three top curves) and bad (three bottom curves) RSSI quality conditions.}\label{fig:phase-II-ESNR-SIR}
	\end{subfigure}
\hspace{0.01\linewidth}	\begin{subfigure}[t]{.48 \linewidth}
		\centering
		\includegraphics[width=1.0 \linewidth]{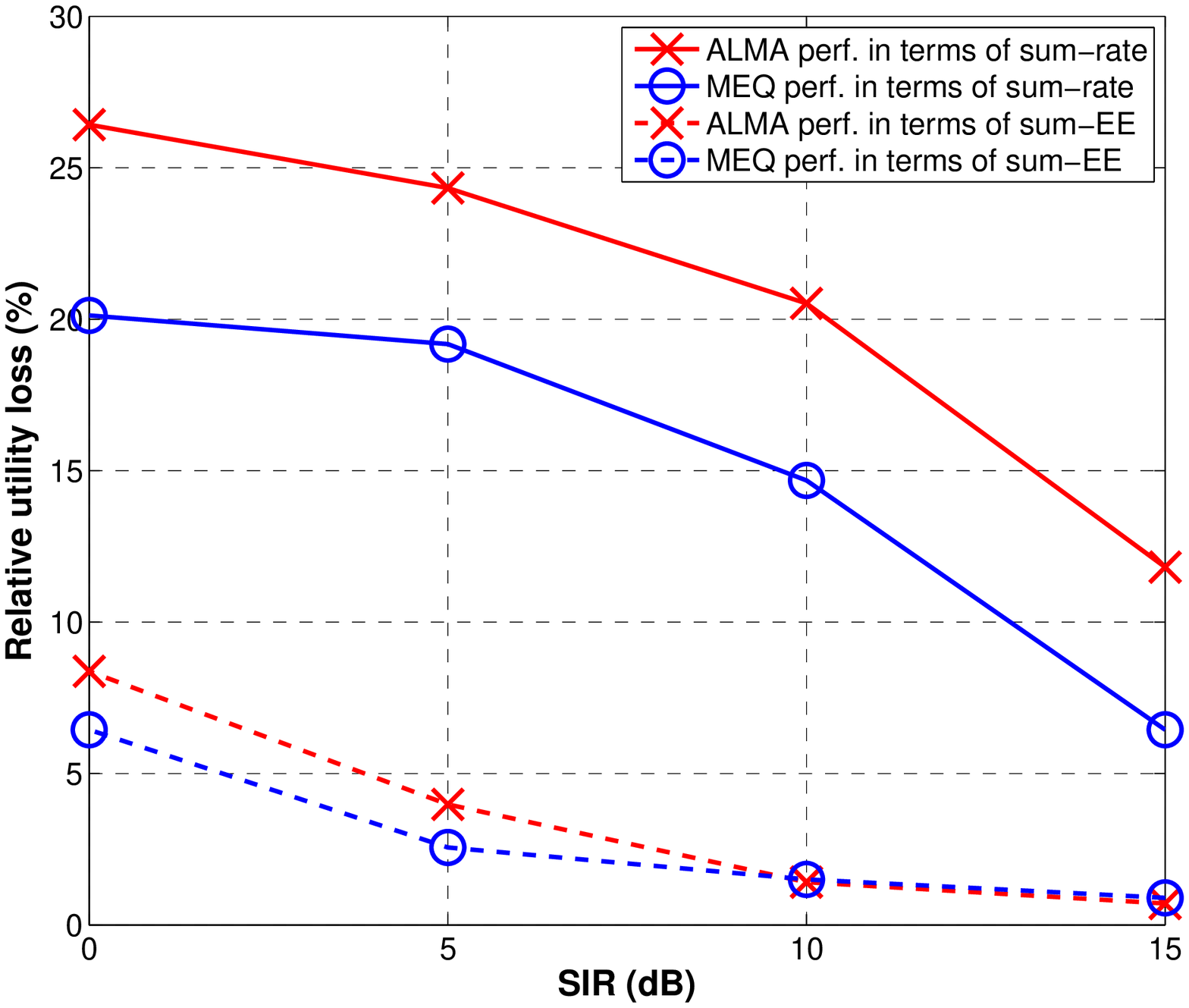}
		\caption{Performance measured by relative utility loss, with utility being the sum-EE or sum-rate.} \label{fig:phase-II-Delta-u-SIR}
	\end{subfigure}
	\caption{\small Performance analysis of conventional LMA, ALMA and MEQ assuming Phase I to be perfect.}
\end{figure}

In this section, we assume Phase I to be perfect. Again, this choice is made to isolate the impact of Phase II estimation techniques on the estimation SNR and the utility functions which are considered for the exploitation phase. When $L=2$ and we quantize with 1-bit, we map the smallest representative of the quantizer to the lowest power and the largest to the highest power level in $\mathcal{P}$ and the other element. If $L>2$, the power levels belong to the set $\left\{0, \frac{1}{L-1}P_{\max},\frac{2}{L-1}P_{\max},...,P_{\max}\right\}$ are picked and the representatives are mapped in the order corresponding to their value. In Phase II, the most relevant techniques to be determined is the quantization of the channel gains estimated through Phase I. 

For $K=2$ users, $S=1$ band, $L=2$ power levels, and $\mathrm{SNR(dB)}=30$, Fig.~\ref{fig:phase-II-ESNR-SIR} provides ESNR(dB) versus SIR(dB) for the three channel gain quantizers mentioned in this paper: ALMA, LMA, and MEQ. The three quantizers are assumed to quantize the channel gains with only $1$ bit. Since only two power levels are exploited over Phase II, this means that the local CSI exchange phase (Phase II) comprises $K$ time-slots. The three top curves of Fig.~\ref{fig:phase-II-ESNR-SIR} correspond to $N=8$ RS power quantization bits and $\epsilon = 1\%$ while the three bottom curves correspond to $N=2$ bits and $\epsilon = 10\%$. First of all, it is seen that the obtained values for ESNR are much lower than for Phase I. Even in the case where $N=8$ and $\epsilon = 1\%$, the ESNR is around $10$ dB whereas it was about $40$ dB for Phase I. This shows that the limiting factor for the global estimation accuracy will come from Phase II; additional comments on this point are provided at the end of this section. Secondly, Fig.~\ref{fig:phase-II-ESNR-SIR} shows the advantages offered by the proposed ALMA over the conventional LMA. 

Fig.~\ref{fig:phase-II-Delta-u-SIR} depicts for $K=2$, $S=1$, $N=8$, $\epsilon= 1 \%$ the average relative utility loss $\Delta u$ in $\%$ against the SIR in dB for ALMA and MEQ. The two bottom (resp. top) curves correspond to the sum-rate (resp. sum-EE). The relative utility loss is seen to be comparable to the one obtained for Phase I. Interestingly, MEQ is seen to induce less performance losses than ALMA, showing that the $\mathrm{ENSR}$ or distortion does not perfectly reflect the need in terms optimality for the exploration phase. This observation partly explains why we have chosen MEQ in Sec.~\ref{sec:global-perf-analysis} for the global performance evaluation; many other simulations (which involve various values for $K$, $N$, $S$, $\epsilon$, etc) not provided here confirm this observation. 


\begin{figure}[h]
	\centering
	\begin{subfigure}[t]{.48 \linewidth}
		\centering
		\includegraphics[width=1.0 \linewidth]{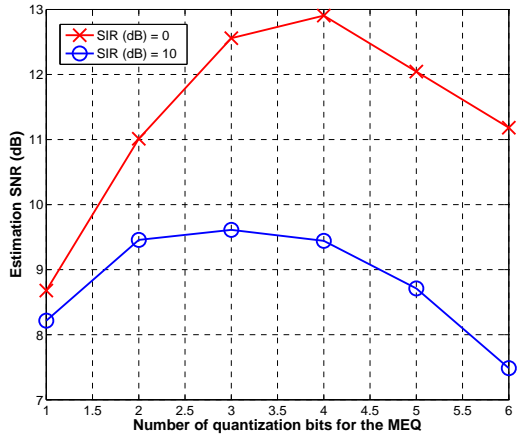}
		\caption{ESNR against quantization bits used in MEQ.} \label{fig:phase-II-ESNR-nb-bits-MEQ}
	\end{subfigure}
\hspace{0.01\linewidth}
	\begin{subfigure}[t]{.48 \linewidth}
		\centering
		\includegraphics[width=1.0 \linewidth]{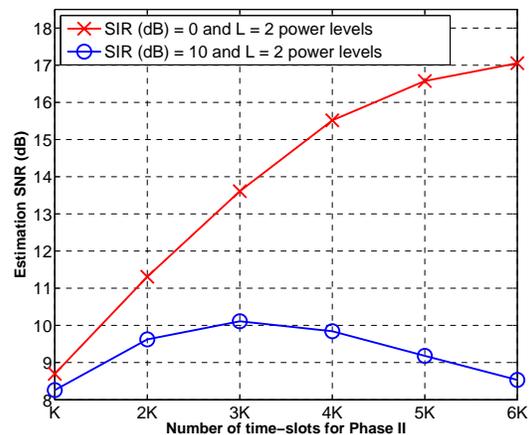}
		\caption{ESNR against $T_\mathrm{II}$}  \label{fig:phase-II-ESNR-nb-time-slots}
	\end{subfigure}
	\caption{\small The power level decoding scheme proposed in this paper is simple and has the advantage of being usable for the SINR feedback instead of RSSI feedback. However, the proposed scheme exhibits a limitation in terms of coordination ability when the inference is very low. The consequence of this is the existence of a maximum ESNR for Phase II. Here we observe that despite increasing the number of quantization bits or time slots used, the ESNR is bounded.}
\end{figure}

\vspace{-0.15cm}
\vspace{-0.15cm}
An important comment made previously is that Phase II constitutes the bottleneck in terms of estimation accuracy for the final global CSI estimate available for the exploitation phase. Here, we provide more details about this limitation. Indeed, even when the quality of the RSSI is good, the ESNR only reaches $10$ dB and even increasing the quantization bits by increasing the power modulation levels or time slots used does not improve the ESNR as demonstrated by the following figures.

For $N=8$ RS power quantization bits and $\epsilon = 1\%$, $\mathrm{SNR(dB)}=30 $, Fig.~\ref{fig:phase-II-ESNR-nb-bits-MEQ} shows the ESNR versus the number of channel quantization bits used by MEQ. It is seen that the ESNR reaches a maximum whether a high interference scenario ($\mathrm{SIR(dB)}=0$) or a low interference scenario ($\mathrm{SIR(dB)}=10$) is considered. In Fig.~\ref{fig:phase-II-ESNR-SIR}, the ESNR was about $9$ dB when the $1-$bit MEQ is used and the SIR equals $0$ dB. Here we retrieve this value and see that the ESNR can reach $13$ dB when the $4-$bit MEQ is implemented, meaning that $16$ power levels are used in Phase II. Now, when the SIR is higher, using the $2-$bit MEQ is almost optimal. If the RSSI quality degrades, then using only $1$ or $2$ bits for MEQ is always the best configuration.

Another approach would be to increase the number of channel gain quantization bits and still only use two power levels over Phase II by increasing the number of time-slots used in Phase II. Fig.~\ref{fig:phase-II-ESNR-nb-time-slots} assumes exactly the same setup as Fig.~\ref{fig:phase-II-ESNR-nb-bits-MEQ} but here it represents the ESNR as a function of the number of time-slots used in Phase II. Here again, an optimal number of time-slots appears for the same reason as for Fig.~\ref{fig:phase-II-ESNR-nb-bits-MEQ}. Both for Fig.~\ref{fig:phase-II-ESNR-nb-bits-MEQ} and Fig.~\ref{fig:phase-II-ESNR-nb-time-slots}, one might wonder why the ESNR is better when the interference is high. This is due to the fact that when the interference is very low, the decoding operation of the power levels of the others becomes less reliable. The existence of maximum points in Fig.~\ref{fig:phase-II-ESNR-nb-bits-MEQ} and Fig.~\ref{fig:phase-II-ESNR-nb-time-slots} precisely translates the tradeoff between the channel gain quantization noise and power level decoding errors. 

\vspace{-0.15cm}
\section{Conclusion}
\label{sec:conclusion}

First, we would like to remind a few comments about the scope and originality of this paper. One of the purposes of this paper is to show that the sole knowledge of the received power or SINR feedback is sufficient to recover global CSI. The proposed technique comprises two phases. Phase I allows each transmitter to estimate local CSI. Obviously, if there already exists a dedicated feedback or signalling channel which allows the transmitter to estimate local CSI, Phase I may be skipped. But even in the latter situation, the problem remains to know how to exchange local CSI among the transmitters. Phase II proposes a completely new solution for exchanging local CSI, namely using power modulation. Phase II is based in particular on a robust quantization scheme of the local channel gains. Phase II is therefore robust against perturbations on the received power measurements; it might even be used for $1-$bit RSSI which would correspond to an ACK/NACK-type feedback, showing that even a rough feedback channel may help the transmitters to coordinate. Note that the proposed technique is general and can be used to exchange and kind of information and not only local CSI.

Second, we summarize here a few observations of practical interest. For Phase I, two estimators have been proposed for Phase I: the LSPD and the MMSEPD estimators. Simulations show that using the MMSEPD requires some statistical knowledge and is more complex, but is well motivated when the RS power is quantized roughly or the feedback channel is very noisy. Otherwise, the use of the LSPD estimator is shown to be sufficient. During Phase II, transmitters exchange local CSI by encoding it onto their power level and using interference as a communication channel; Phase II typically requires $K$ time-slots at least (assuming all transmitters simultaneously communicate in Phase II), which makes $2K$ time-slots for the whole estimation procedure. This is typically the number of time-slots needed by IWFA to converge, when it converges. For Phase II, three estimation schemes are provided which are in part based on one of the two quantizers ALMA and MEQ; the quantizers are computed offline but are exploited online. MEQ seems to offer a good trade-off between complexity and performance in terms of sum-rate or sum-energy-efficiency. In contrast with Phase I in which the estimation SNR typically reaches $40$ dB for good RS power measurements, the estimation SNR in Phase II is typically around $10$ dB, showing that Phase II will constitute the bottleneck in terms of estimation quality of global CSI. This is due to fact that the cross channel gains may be small when they fluctuate (this would not occur in the presence of Rician fading), which generates power level decoding errors. As explained, one way of improving the estimation SNR over Phase II is to activate only one user at a time, but then the proposed power level decoding scheme would only apply to RSSI feedback and not to SINR feedback anymore. In Phase III, having global CSI, each transmitter can apply the BRD to the sum-utility instead of applying it to an individual utility as IWFA does, resulting in a significant performance improvement as seen from our numerical results.

\vspace{-0.15cm}
\vspace{-0.15cm}
\vspace{-0.15cm}
\appendices

\section{proof of proposition III.1}
{\bf Proof: } From Section II, we have $\ul{\widehat{\omega}}_{i} \in \Omega$ and  $\ul{\widetilde{\omega}}_{i} \in \Omega$, where $\Omega$ is a discrete set. Therefore, we can rewrite the likelihood probability $\Pr\left(\widetilde{\underline{\omega}}_{i}|\underline{g}_{i}\right)$ as follows
\vspace{-0.15cm}
\begin{equation}
\begin{array}{rl}
\Pr\left(\widetilde{\underline{\omega}}_{i}|\underline{g}_{i}\right) 
 & \stackrel{\left(a\right)}{=}   \underset{m=1}{\overset{M^{T_{\mathrm{I}}}}{\sum}}\Pr\left(\widetilde{\underline{\omega}}_{i}|\widehat{\underline{\omega}}_{i}=\ul{\mathrm{w}}_m\right)\Pr\left(\widehat{\underline{\omega}}_{i}=\ul{\mathrm{w}}_m|\underline{g}_{i}\right) \\
& \stackrel{\left(b\right)}{=}  \underset{m=1}{\overset{M^{T_{\mathrm{I}}}}{\sum}}\Pr\left(\widehat{\underline{\omega}}_{i}=\ul{\mathrm{w}}_m|\underline{g}_{i}\right)\overset{T_{\mathrm{I}}}{\underset{t=1}{\prod}}\Gamma\left(\widetilde{\omega}_{i}\left(t\right)|\widehat{\omega}_{i}\left(t\right)\right) \\
&\stackrel{\left(c\right)}{=}   \overset{T_{\mathrm{I}}}{\underset{t=1}{\prod}}\Gamma\left(\widetilde{\omega}_{i}\left(t\right)|Q_{\mathrm{RS}}\left(\ul{e}_t^{\mathrm{T}}   \mathbf{P}_{ \mathrm{I} } \underline{g}_{i}+\sigma^2\right)\right) \label{eq:ML2}
\end{array}\end{equation} where $\ul{e}_t$ is a column vector whose entries are zeros except for the $t^{\mathrm{th}}$. In (\ref{eq:ML2}), (a) holds as the estimation and feedback process $\underline{g}_{i}$ to $\widehat{\underline{\omega}}_{i}$ to $\widetilde{\underline{\omega}}_{i}$ (represented in Fig. 1) is Markovian, (b) holds because the DMC is separable and (c) holds because $\Pr\left(\widehat{\underline{\omega}}_{i}|\underline{g}_{i}\right)$ is a discrete delta function that is zero everywhere except when $Q_{\mathrm{RS}}\left( \mathbf{P}_{ \mathrm{I}} \underline{g}_{i}\right)=\widehat{\underline{\omega}}_{i}$. 

From (\ref{eq:ML2}), the set of the ML estimators can now be written as
\begin{equation}
\mathcal{G}_i^{\mathrm{ML}}=\left\{\arg \underset{\underline{g}_{i}}{\max} \overset{T_{1}}{\underset{t=1}{\prod}}\Gamma\left(\widetilde{\omega}_{i}\left(t\right)|Q_{\mathrm{RS}}\left(\ul{e}_t^{\mathrm{T}}   \mathbf{P}_{ \mathrm{I} } \underline{g}_{i}+\sigma^2\right)\right) \right\}
\label{eq:ML_final}
\end{equation} which is the first claim of our proposition. Now, we look at the LS estimator, which is know from (\ref{eq:LS}) to be
\begin{equation}
 \mathbf{P}_{ \mathrm{I}} \underline{g}_{i}^{\mathrm{LSPD}}+ \sigma^2 \ul{1}=\underline{\widetilde{\omega_{i}}}
\end{equation}
or equivalently:
\begin{equation}
\ul{e}_t^{\mathrm{T}}   \mathbf{P}_{ \mathrm{I} } \underline{g}_{i}^{\mathrm{LSPD}}+\sigma^2=\widetilde{\omega_{i}}\left(t\right)
\label{eq:sg}
\end{equation}


If for all $\ell$, $\arg \underset{k}{\max} \:\:\Gamma(\mathrm{w}_\ell |  \mathrm{w}_k )   =   \ell$, then the ML set can be evaluated based on (\ref{eq:ML_final}) as
\begin{equation}
\mathcal{G}_i^{\mathrm{ML}}=\left\{\underline{g}_{i}| \forall t, Q_{\mathrm{RS}}\left(\ul{e}_t^{\mathrm{T}}   \mathbf{P}_{ \mathrm{I} } \underline{g}_{i}+\sigma^2\right)=\widetilde{\omega_{i}}\left(t\right) \right\}
\label{eq:ML3}
\end{equation}
Therefore, we observe that if $\mathcal{G}_i^{\mathrm{ML}}$ is given as in (\ref{eq:ML3}), then from (\ref{eq:sg}), we have $ \underline{g}_{i}^{\mathrm{LSPD}} \in \mathcal{G}_i^{\mathrm{ML}}$, our second claim. \hfill $\blacksquare$

\section{proof of proposition III.2}
{\bf Proof: }
After the RSSI quantization, the $M^{T_{\mathrm{I}}}$ different levels of $\ul{\widehat{\omega}}_{i}$ or  $\ul{\widetilde{\omega}}_{i}$ are $\ul{\mathrm{w}}_{1},\ul{\mathrm{w}}_{2},..,\ul{\mathrm{w}}_{M^{T_{\mathrm{I}}}}$ forming the set $\Omega$. 

Define by $h:\Omega \to G$ which maps the observed RSSI feedback to a channel estimate, where $G:= \{\ul{\mathrm{g}}_1,\ul{\mathrm{g}}_2,...,\ul{\mathrm{g}}_{M^{T_{\mathrm{I}}}} \}$, such that $h(\ul{\mathrm{w}}_{m}) = \ul{\mathrm{g}}_{m}$. That is, when transmitter $i$ observes the RSSI feedback $\ul{\widetilde{\omega}}_{i}$ to be $\ul{\mathrm{w}}_{m}$, local channel estimate $\ul{\widetilde{g}}_i$ is $\ul{\mathrm{g}}_{m}$.

Based on the above definitions, we have that
\begin{equation}\begin{split}
\mathbb{E}\left[|\underline{\widetilde{g}}_i-\underline{g}_i|^2\right]
	=	{\sum_{n=1}^{M^{T_{\mathrm{I}}}}}\int_{\underline{{x}} \in \mathbb{R}_{\geq 0}^K } \Pr\left(\underline{\widetilde{g}}_i=\ul{\mathrm{g}}_{n}|\underline{{g}}_i=\ul{x}\right)\phi_{i}\left(\underline{{x}}\right)|\ul{\mathrm{g}}_{n}-\underline{x}|^2{\mathrm{d}\underline{{x}}}
	\label{eq:mmseproof1}
		\end{split}
	\end{equation}
The term $  \Pr\left(\underline{\widetilde{g}}_i=\ul{\mathrm{g}}_{n}|\underline{{g}}_i=\ul{x}\right)$ can be further expanded as
\begin{equation}
\begin{array}{rl}
 \Pr\left(\underline{\widetilde{g}}_i=\ul{\mathrm{g}}_{n}|\underline{{g}}_i=\ul{x}\right)  = \sum_{\ell=1}^{M^{T_{\mathrm{I}}}}\sum_{m=1}^{M^{T_{\mathrm{I}}}} \Pr\left(\underline{\widetilde{g}}_i=\ul{\mathrm{g}}_{n},\underline{\widetilde{\omega}}_i=\ul{\mathrm{w}}_\ell,\underline{\widehat{\omega}}_i=\ul{\mathrm{w}}_m|\ul{g}_i=\ul{x}\right) \\
 = \sum_{\ell=1}^{M^{T_{\mathrm{I}}}}\sum_{m=1}^{M^{T_{\mathrm{I}}}} \Pr\left(\underline{\widetilde{g}}_i=\ul{\mathrm{g}}_{n}|\underline{\widetilde{\omega}}_i=\ul{\mathrm{w}}_\ell \right) \Pr\left(\underline{\widetilde{\omega}}_i=\ul{\mathrm{w}}_\ell|\underline{\widehat{\omega}}_i=\ul{\mathrm{w}}_m\right)
 \Pr\left(\underline{\widehat{\omega}}_i=\ul{\mathrm{w}}_m|\ul{g}_i=\ul{x}\right) 
\end{array}
\label{eq:profcomb1}
\end{equation}

Now we know that the mapping $h()$ is deterministic and results in $h(\ul{\mathrm{w}}_{m}) = \ul{\mathrm{g}}_{m}$. Therefore, $\Pr\left(\underline{\widetilde{g}}_i=\ul{\mathrm{g}}_{n}|\underline{\widetilde{\omega}}_i=\ul{\mathrm{w}}_\ell \right) = \delta_{n,\ell}$, where $\delta_{n,\ell}$ is the Kronecker delta function such that $\delta_{n,\ell} = 0$ when $n \neq \ell$ and $\delta_{n,\ell}=1$ when $n=\ell$. Additionally, we also know that $\Pr\left(\underline{\widetilde{\omega}}_i=\ul{\mathrm{w}}_\ell|\underline{\widehat{\omega}}_i=\ul{\mathrm{w}}_m\right) =\prod_{t=1}^{T_I} \Gamma\left(\ul{\mathrm{w}}_\ell(t)|\ul{\mathrm{w}}_m(t)\right)$ by definition (where $\ul{\mathrm{w}}_m(t)$ is the $t$-th component of $\ul{\mathrm{w}}_m$) . This results in (\ref{eq:profcomb1}) being simplified to
\begin{equation}
 \Pr\left(\underline{\widetilde{g}}_i=\ul{\mathrm{g}}_{n}|\underline{{g}}_i\right) = \sum_{m=1}^{M^{T_{\mathrm{I}}}} \prod_{t=1}^{T_I} \Gamma\left(\ul{\mathrm{w}}_n(t)|\ul{\mathrm{w}}_m(t)\right) \Pr\left(\underline{\widehat{\omega}}_i=\ul{\mathrm{w}}_m|\ul{g}_i=\ul{x}\right)
\label{eq:profcomb2} 
\end{equation}
Recall that $\underline{\widehat{\omega_{i}}} = Q_{\mathrm{RS}}\left( \mathbf{P}_{ \mathrm{I}} \underline{g_{i}}\right)$ by definition of the quantizer. Define by 
\begin{equation}
\mathcal{G}_{m}:=\left\{\underline{x} \in \mathbb{R}_{\geq 0}^K: \     Q_{\mathrm{RS}}\left(\mathbf{P}_{ \mathrm{I} } \underline{x}+\sigma^2\ul{1}\right) =\ul{\mathrm{w}}_m \right\}
\end{equation} resulting in
\begin{equation}
\Pr \left(\underline{\widehat{\omega}}_i=\ul{\mathrm{w}}_m|\ul{g}_i=\ul{x}\right) = \left\{ \begin{array}{ll}
1 & \mathrm{if} \,\, \ul{x} \in  \mathcal{G}_{m}\\
0 & \mathrm{if} \,\, \ul{x} \notin  \mathcal{G}_{m}
\end{array}\right.
\label{eq:probofquant}
\end{equation}
Now, we can simplify (\ref{eq:mmseproof1}) using (\ref{eq:probofquant}) and (\ref{eq:profcomb2}) into
\begin{equation}\begin{split}
&\mathbb{E}\left[|\underline{\widetilde{g}}_i-\underline{g}_i|^2\right] \\
	=	&{\sum_{n=1}^{M_{T_{\mathrm{I}}}}}{\sum_{m=1}^{M_{T_{\mathrm{I}}}}} \overset{T_{1}}{\underset{t=1}{\prod}}\Gamma\left(\ul{\mathrm{w}}_n(t)|\ul{\mathrm{w}}_m(t)\right) \int_{\mathcal{G}_{m}}\phi_{i}\left(\ul{x}\right)|\ul{\mathrm{g}}_{n}-\ul{x}|^2{\mathrm{d}\ul{x}}
	\label{eq:mmseproof2}
		\end{split}
	\end{equation}

For a fixed DMC, we can find the $\ul{{g}}_{i}^{\mathrm{MMSE}}$ which will minimize the distortion by taking the derivative of the distortion over $\ul{\mathrm{g}}_{n}$: 
\begin{equation}
\begin{split}
\frac{\partial \mathbb{E}\left[|\underline{\widetilde{g}}_i-\underline{g}_i|^2\right]}{\partial \ul{\mathrm{g}}_{n} }=2{\sum_{m=1}^{M_{T_{\mathrm{I}}}}} \overset{T_{1}}{\underset{t=1}{\prod}}\Gamma\left(\ul{\mathrm{w}}_n(t)|\ul{\mathrm{w}}_m(t)\right)\int_{\mathcal{G}_{m}}\phi_{i}\left(\ul{x}\right)\left(\ul{\mathrm{g}}_{n}-\underline{x}\right){\mathrm{d}\underline{{x}}}\label{eq:deri}
\end{split}
\end{equation}
To minimize distortion, this derivative should be equal to zero. The $\ul{\mathrm{g}}_{n}$ minimizing the distortion is by definition, the MMSE of the channel given $\ul{\widetilde{\omega}}_i = \ul{\mathrm{w}}_{n}$. Therefore by rearranging (\ref{eq:deri}), we can find the expression for the MMSE given in the proposition III.2. \hfill $\blacksquare$
%

\vspace{-0.15cm}
\section{Calculations for the ALMA}
As defined in the main text, $\widetilde{g}_{ji}^{k} \in \{v_{ji,1},...,v_{ji,R}\}$ and the p.d.f. of $\widetilde{g}_{ji}$ is denoted by $\gamma_{ji}$ in general. Note that when $\widetilde{g}_{ji}$ belongs to a discrete set, we can replace the integrals and $\gamma_{ji}$ with a sum and discrete probability function without any significant alteration to our results and calculations. Denoting the p.d.f of $g_{ji}$ by $\phi_{ji}$, the distortion between $g_{ji}$ and $\widetilde{g}_{ji}^{k}$ can be written as
\vspace{-0.15cm}
{\begin{equation}\begin{split}
\mathbb{E}[\left(g_{ji}-\widetilde{g}_{ji}^{k}\right)^{2}]
	=	\sum_{r=1}^{R}\underset{x,\widetilde{x} \in \mathbb{R}_{\geq 0}}{\int} \Pr\left(\widetilde{g}_{ji}^{k}=v_{ji,r}|\widetilde{g}_{ji}=\widetilde{x}\right)\gamma_{ji}\left(\widetilde{x}|x\right)\phi_{ji}\left(x\right)
	\left(x-v_{ji,r}\right)^{2}\mathrm{d}x \mathrm{d}\widetilde{x}\\
	\label{eq:ac1}
	\end{split}\end{equation}
	\vspace{-0.15cm}
	 which is the distortion observed by transmitter $k$ when transmitter $i$ communicates $g_{ji}$ in Phase II.
	 As the transmitter $i$ estimates $g_{ji}$ as $\widetilde{g}_{ji}$, the quantization operation $Q_i^{\mathrm{II}}$ is performed resulting in $\widetilde{g}_{ji}$ being quantized into a certain representative $v_{ji,n}$, if $\widetilde{g}_{ji} \in [u_{ji,n},u_{ji,n+1})$. Given that the transmitter $i$ operates at a power level corresponding to $v_{ji,n}$, the transmitter $k$ will decode $v_{ji,r}$ with a probability $\pi(r|n)$ as defined in Section IV. Now we can expand the term $\Pr\left(\widetilde{g}_{ji}^{k}=v_{ji,r}|\widetilde{g}_{ji}=\widetilde{x}\right)$ in the following manner.
\vspace{-0.15cm}
\begin{equation}\begin{split}
\Pr\left(\widetilde{g}_{ji}^{k}=v_{ji,r}|\widetilde{g}_{ji}=\widetilde{x}\right) 
&= \sum_{n=1}^R \Pr\left(\widetilde{g}_{ji}^{k}=v_{ji,r}|Q_i^{\mathrm{II}}(\widetilde{g}_{ji})=v_{ji,n}\right) \Pr\left(Q_i^{\mathrm{II}}(\widetilde{g}_{ji})=v_{ji,n}|\widetilde{g}_{ji}=\widetilde{x}\right)\\
&= \sum_{n=1}^R \pi(r|n)  \Pr\left(Q_i^{\mathrm{II}}(\widetilde{g}_{ji})=v_{ji,n}|\widetilde{g}_{ji}=\widetilde{x}\right)
\end{split}
\label{eq:probexpand}
\end{equation} where we know
\begin{equation}
\Pr\left(Q_i^{\mathrm{II}}(\widetilde{g}_{ji})=v_{ji,n}|\widetilde{g}_{ji}=\widetilde{x}\right) = \left\{ \begin{array}{ll}
1 & \mathrm{if} \,\, \widetilde{x} \in  [u_{ji,n},u_{ji,n+1})\\
0 & \mathrm{if} \,\, \widetilde{x} \notin  [u_{ji,n},u_{ji,n+1})
\end{array}\right.
\label{eq:probofquant2}
\end{equation}

Substituting (\ref{eq:probofquant2}) and (\ref{eq:probexpand}) in (\ref{eq:ac1}), we get
\begin{equation}\begin{split}
 \mathbb{E}[\left(g_{ji}-\widetilde{g}_{ji}^{k}\right)^{2}] 
	=\sum_{n=1}^{R}{\sum_{r =1}^R} \pi_{ji} \left(r |n\right) \underset{x = 0}{\overset{\infty}{\int}} \hspace{.2cm}\underset{\widetilde{x} = u_{ji,n}} {\overset{u_{ji,n+1}}{\int}} 
\gamma_{ji}\left(\widetilde{x}|x\right)\phi_{ji}\left(x\right)\left(x-v_{ji,r}\right)^{2} \mathrm{d}x \mathrm{d}\widetilde{x}.
\end{split}\end{equation}

For fixed transition levels $u_{ji,n}$, the optimum representatives $v_{ji,r'}$ are obtained by setting the partial derivatives of the distortion $E[\left(g_{ji}-\widetilde{g}_{ji}^{k}\right)^{2}]$, with respect to $v_{ji,r'}$, to zero. That is
\[\begin{split}
\frac{\partial {\mathbb{E}[\left(g_{ji}-\widetilde{g}_{ji}^{k}\right)^{2}]}}{\partial v_{ji,r'}} =  {\sum_{n=1}^R}\pi_{ji}\left(r' |n\right)
 \int_{x=0}^{\infty}\int_{\widetilde{x}=u_{ji,n}}^{u_{ji,n+1}} 2\gamma_{ji}( \widetilde{x} | x  )  \phi_{ji}(x)\left(x-v_{ji,r'}\right)\mathrm{d}x \mathrm{d}\widetilde{x}
= 0 
\end{split}\]
which results in
\begin{equation}
v_{ji,r'}  =\frac{{\displaystyle\sum_{n=1}^R}\pi_{ji} \left(r' |n\right)\displaystyle\int_{x=0}^{\infty}
\displaystyle\int_{ \widetilde{x}= u_{ji,n}  }^{u_{ji,n+1}  } x
\gamma_{ji}( \widetilde{x} |x  )  \phi_{ji}(x)
\mathrm{d}\widetilde{x} \mathrm{d}x }{{\displaystyle\sum_{n=1}^R} \pi_{ji} \left(r' |n\right) \displaystyle\int_{x=0}^{\infty}\displaystyle\int_{\widetilde{x}=u_{ji,n} }^{u_{ji,n+1}}
\gamma_{ji}(\widetilde{x} |x  )  \phi_{ji}(x)
\mathrm{d}\widetilde{x} \mathrm{d}x }. 
\label{eq:ALMr}
\end{equation} 
For fixed representatives $v_{ji,r}$, the optimum transition levels $u_{ji,n'}$ are obtained by setting the partial derivatives of the distortion $E[\left(g_{ji}-\widetilde{g}_{ji}^{k}\right)^{2}]$ with respect to $u_{ji,n'}$, to zero. We use the second fundamental theorem of calculus, i.e., $\frac{\mathrm{d} }{\mathrm{d}x} \int_a^x f(t) \mathrm{d}t = f(x)$ to obtain $u_{ji,n'}$ for all $n' \in \{2,..,R\}$ as
\begin{equation}\begin{split}
\frac{\partial {\mathbb{E}[\left(g_{ji}-\widetilde{g}_{ji}^{k}\right)^{2}]}}{\partial u_{ji,n'}}
={\sum_{r=1}^R} (\pi_{ji}\left(r|n'-1\right)-\pi_{ji}\left(r|n'\right))\int_{0}^{\infty}\gamma_{ji}(u_{ji,n'}   | x  )  \phi_{ji}(x)\left({v_{ji,r}-x}\right)^2 \mathrm{d}x 
=0
\label{eq:ALM}
\end{split}\end{equation}
with $u_{ji,1}=0$ and $u_{ji,R+1}= \infty$ as the boundary conditions. Solving the above conditions is very difficult as the variable to solve is inside the integral as an argument of $\gamma$. Therefore we consider the special case where $\gamma_{ji}(\widehat{x}|x) = \delta(x-\widehat{x}) $ where $\delta$ is the Dirac delta function which is $0$ at all points except at $0$ and whose integral around a neighborhood of $0$ is $1$. This corresponds to the case where the channel is perfectly estimated after phase I. This directly transforms (\ref{eq:ALMr}) to (\ref{eq:alm_r1}) of the ALMA, and we can simplify (\ref{eq:ALM}) into
\vspace{-0.15cm}
\begin{equation}\begin{split}
0 = {\sum_{r=1}^R} \left[\pi_{ji}\left(r|n'-1\right)-\pi_{ji}\left(r|n'\right) \right] \phi_{ji}(u_{ij,n'})\left(v_{ji,r}-u_{ij,n'}\right)^2
\label{eq:ALM1}
\end{split}\end{equation}
We have ${\sum_{r=1}^R} \left[\pi_{ji}\left(r|n'-1\right)-\pi_{ji}\left(r|n'\right) \right] \left(u_{ij,n'}\right)^2=0$ since ${\sum_{r=1}^R} \pi_{ji}\left(r|n'\right) =1$, resulting in 
\begin{equation}\begin{split}
u_{ij,n'} = \frac{{\sum_{r=1}^R} \left[\pi_{ji}\left(r|n'-1\right)-\pi_{ji}\left(r|n'\right) \right] v_{ji,r}^2 }{2 {\sum_{r=1}^R} \left[\pi_{ji}\left(r|n'-1\right)-\pi_{ji}\left(r|n'\right) \right]v_{ji,r} }
\label{eq:ALM2}
\end{split}\end{equation}
which is (\ref{eq:alm_t1}) used in the ALMA. \hfill $\blacksquare$

\def\bibfont{\footnotesize}
\bibliographystyle{unsrt}

\end{document}